\documentclass[a4paper,11pt]{article}
\usepackage[margin=1in,footskip=0.25in, a4paper]{geometry}

\usepackage[english]{babel}
\usepackage{amsmath,amssymb, enumitem}
\usepackage{bbm, dsfont}
\usepackage{tikz}
\usepackage{natbib}

\usepackage[utf8]{inputenc}
\usepackage{hyperref}
\usepackage{graphicx}
\usepackage{authblk}

 \usepackage{titlesec}

 \usepackage{array}
 \usepackage{graphicx}
 \usepackage[space]{grffile}
 \usepackage{latexsym}
 \usepackage{textcomp}
 \usepackage{longtable}
 \usepackage{tabulary}
 \usepackage{booktabs,array,multirow}
\usepackage{amsfonts,amsmath,amssymb}
 \usepackage{subfigure}

\RequirePackage{amsthm,amsmath,amsfonts,amssymb,amstext}

\newcommand{\einsfun}{\mathbf 1}
\newcommand{\D}{\mathrm{d}}

\renewenvironment{abstract}
  {{\bfseries\noindent{\abstractname}\par\nobreak}\footnotesize}
  {\bigskip}

\titlespacing{\section}{0pt}{*3}{*1}
\titlespacing{\subsection}{0pt}{*2}{*0.5}
\titlespacing{\subsubsection}{0pt}{*1.5}{0pt}

\theoremstyle{plain}
\newtheorem{theorem}{Theorem}

\theoremstyle{remark}
\newtheorem{example}[theorem]{Example}

\AtBeginDocument{\DeclareGraphicsExtensions{.pdf,.PDF,.eps,.EPS,.png,.PNG,.tif,.TIF,.jpg,.JPG,.jpeg,.JPEG}}

\begin{document}

\title{Asymmetric dependence in hydrological extremes}

\author[1]{Cristina Deidda\thanks{corresponding author (cristina.deidda@vub.be)}}%
\author[2]{Sebastian Engelke}%
\author[3]{Carlo De Michele}
\affil[1]{Department of Hydrology and Hydraulic Engineering, Vrije Universiteit Brussel, Brussels, Belgium}%
\affil[2]{Research Center for Statistics, University of Geneva, Geneva, Switzerland}%
\affil[3]{Department of Civil and Environmental Engineering, Politecnico di Milano, Milano, Italy}%

\vspace{-1em}
\date{}

\begingroup
\let\center\flushleft
\let\endcenter\endflushleft
\maketitle
\endgroup

\begin{abstract}
Extremal dependence describes the strength of correlation between the largest observations of two variables.
It is usually measured with symmetric dependence coefficients that do not depend on the order of the variables. In many cases, there is a natural asymmetry between extreme observations that can not be captured by such coefficients.
An example for such asymmetry are large discharges at an upstream and a downstream stations on a river network: an extreme discharge at the upstream station will directly influence the discharge at the downstream station, but not vice versa.
Simple measures for asymmetric dependence in extreme events have not yet been investigated. 
We propose the asymmetric tail Kendall's $\tau$ as a measure for extremal dependence that is sensitive to asymmetric behaviour in the largest observations. It essentially computes the classical Kendall's $\tau$ but conditioned on the extreme observations of one of the two variables. We show theoretical properties of this new coefficient and derive a formula to compute it for existing copula models. We further study its effectiveness and connections to causality in simulation experiments. 
We apply our methodology to a case study on river networks in the United Kingdom to illustrate the importance of measuring asymmetric extremal dependence in hydrology. Our results show that there is important structural information in the asymmetry that would have been missed by a symmetric measure. Our methodology is an easy but effective tool that can be applied in exploratory analysis for understanding the connections among variables and to detect possible asymmetric dependencies.

\end{abstract}

\section*{Plain Language Summary}
Compound events describe situations where the simultaneous behaviour of two or more variables lead to severe impacts. For instance, the dependence between climate or hydrological variables can lead to particular conditions that result in extreme events at the same time in different locations. Since the physical processes behind these phenomena are very complex, there can be a stronger influence from one variable on another than the other way around. In such cases, there is asymmetry in the dependence between the extreme observations of the two variables. The traditional measures of dependence are symmetric and can not detect any asymmetries. We propose a new measure that is sensitive to asymmetric behaviour in extremes. It is based on an extension of the  Kendall's $\tau$ coefficient, a classical dependence measure. We derive evidence from theory and simulation experiments for the effectiveness of our new methodology.  
We then apply it to a case study on river networks in the United Kingdom where we show that our measure detects asymmetric behaviour of extreme discharges with a preferred direction from upstream to downstream stations. Our work points out the importance of considering proper tools for analyzing the connections between different variables in particular in the presence of asymmetry in extreme observations.

\section{Introduction}

In the last decades, increasing attention has been put on the study of complex extreme events generated by the inter-correlation of multiple variables.
Extreme events at regional scale, such as floods, droughts, or heatwaves, can be caused by the interactions of climate/hydrological variables, or can be influenced and amplified by large scale climate patterns such as the El Niño–Southern Oscillation \citep{Ward2014}.
Compound climate or weather events are complex events generated by multiple variables (or hazards) that interact across multiple spatial and temporal scales \citep{Zscheischler20}. 
An example of flooding due to a compound event is the occurrence of heavy precipitation on saturated soil \citep{KIM2019100629,Fang2016} or on frozen ground \citep{Zhang17}. Similarly, the co-occurrence of heavy precipitation and storm surge increases the coastal flooding probability \citep{Bevacqua2017,Ward2018}. Temporal and spatial clustering of precipitation is shown to increase the flood risk in several case studies around the world \citep{Boers2014natcom,Schneeberger2018,Barton2016,Villarini13,Bevacquaetal2021,Banfi22}. The complexity of studying these problems comes from the fact that the involved variables are often dependent or even exhibit a causal connection. 

The assessment and quantification of the strength of dependence between two variables is often done through summary coefficients. One of the most commonly used summaries is the Pearson correlation $\rho_P$ that measures linear dependence. Other measures of association like the Kendall's $\tau$, the Spearman's $\rho_S$, have been proposed and show slightly different properties, such the invariance under marginal transformations. Also in climate science, the use of such measures is increasingly popular in the analysis of climate or weather hazards (\cite{DepZsch17cop,Runge2019}; \cite{Runge2019b}) and possible connections among the main drivers.

In the study of multivariate extreme events, dependence measures are widely applied in the exploratory analysis and model assessment. In order to focus on the correlation between the most extreme observations of the two components of a bivariate random vector $(X,Y)$, numerous coefficients have been introduced in the extreme value literature. The most commonly used measure, both in theory and applications, is the extremal coefficient \citep{col1999}.
It can be defined as the correlation of large exceedances, that is,
\begin{align}\label{chi}
\chi_{XY} =  \lim_{q\to 1} \text{Cor}\left(\einsfun\{X > F_X^{-1}(q)\}, \einsfun\{Y > F_Y^{-1}(q)\}\right), 
\end{align}
where the marginal distribution functions $F_X$ and $F_Y$ of $X$ and $Y$, respectively, are assumed to be continuous. The indicator function $\einsfun\{X > F_X^{-1}(q)\}$ is one if $X$ exceeds its $q$-quantile $F_X^{-1}(q)$ and zero otherwise, and similarly for $Y$. The coefficient $\chi_{XY}$ therefore measures the correlation between the occurrence of extreme events in $X$ and $Y$, and has for instance been used to study spatial extremes \citep{Davison12,PADOAN20131}. There are many other measures of extremal dependence with similar properties \citep[e.g.,][]{coo2006,lar2012}; see \cite{eng2021} for a review.

Most measures for extremal dependence are symmetric and do not depend on the order of the variables; for instance, the extremal correlation satisfies $\chi_{XY} = \chi_{YX}$.
This means that such measures may capture the overall strength of extremal dependence, but they do not reveal information on the possible \emph{asymmetry} or \emph{direction} of this dependence.
Asymmetries in extremal dependence can arise for many reasons, and several statistical models have been proposed \citep[e.g.,][]{PADOAN2011977}. In hydrology, for instance, two stations on a river network can be affected more or less strongly by certain precipitation events.  
On the other hand, a preferred direction of dependence between the two variables $X$ and $Y$ may also lead to an asymmetric behavior. Such directionality should be understood in the context of a causal relation. In the hydrological example, a downstream station will be causally affected by another station upstream on the river and the dependence between large events will thus typically exhibit asymmetry.

From a theoretical point of view, asymmetric dependence measures have extensively been studied in the context of statistical modeling of multivariate data \citep{Wolfgang21}. 
  In practical applications they are used in exploratory data analysis to guide subsequent modeling choices, for instance in finance and economics \citep{Patton04,Tsafack7,okimoto2008,Hatherley07}.
 In particular, to describe the whole dependence structure, a family of copula models is typically selected. The type of model can be symmetric or asymmetric, depending on the prior information from the dependence coefficients \citep{Durante10,Liebscher08,Kim14}.
 Asymmetric dependence measures can also be used to detect causal directions between the variables \citep{Hatemi12}.
 Beyond economical and financial applications, the importance of asymmetric dependence has been acknowledged in geotechnical analyses \citep{ZHANG20191960} and post processing of weather forecasts \citep{Li2020}.
 In many applications, the focus of the statistical analysis is on the dependence between the most extreme observations. In climate science, so-called compound events may lead to the most severe impacts \citep{Zscheischler18,Ridder20}. Similarly, in hydrology, the risk of the concomitant flood events on a river network attracts increasing attention \citep{asadi2015extremes, DeLuca2017, Kemter20, deidda2021causes}.
 Only few asymmetric dependence measures that are adapted to extremes have been proposed. Most of them are designed to detect causal directions under certain model classes \citep{Mhalla2019, gnecco2021causal, tra2021}. These measures are however not sensitive to more subtle asymmetries in the correlation structure of extremes.

In this work, we explore the asymmetry in the dependence of extremes, both from a theoretical and a practical perspective.
We introduce a new \emph{conditional} version of Kendall's $\tau$ coefficient for extreme observations. This results in an asymmetric extremal dependence coefficient that we call the \emph{asymmetric tail Kendall's $\tau$}. We investigate its theoretical properties and show that it is a natural object to study existing extreme value models.
In particular, we derive a closed form representation of the coefficient for certain models and study its relation to causal structures.
To assess the strength and possible limitations of our approach, we conduct extensive numerical experiments. 
We illustrate the developed methodology on the analysis of extreme floods in eighteen catchments located in the United Kingdom.

\section{Methodology: asymmetric tail Kendall's $\tau$}\label{sec:methods}

\subsection{Definition}\label{sec:def}
Let $(X,Y)$ be a bivariate random vector with continuous marginal distribution functions $F_X$ and $F_Y$, respectively. The Kendall's $\tau$, or Kendall's rank correlation coefficient, is defined as 
\begin{equation}\label{OriginalKT}
 \tau(X,Y) = 2\mathbb P( (X- \tilde X)(Y - \tilde Y) > 0) -1
\end{equation}
where $(\tilde X, \tilde Y)$ is an independent vector identically distributed as $(X,Y)$. The Kendall's $\tau$ takes values in the interval $[-1,+1]$ and measures the association between both variables. The original Kendall's $\tau$ has two drawbacks when applied to data on climatological or hydrological extreme events. Firstly, it is symmetric in the two variables, that is, $\tau(X,Y) = \tau(Y,X)$, and it therefore does not allow the identification of possible asymmetries in the dependence or the causal structure. In hydrological studies involving the river discharge or the level of pollutants at different gauging stations, such asymmetry does often arise. Secondly, Kendall's $\tau$ focuses on the bulk of the distributions of the variables and therefore does not measure accurately the strength of dependence between extreme observations of the two variables. 

We next introduce a new pair of coefficients that can capture possible asymmetries and put the focus on extremes. 
Let $q \in (0,1)$ be a probability level that is typically chosen to be close to one. We define the \emph{(pre-limit) asymmetric  tail Kendall's $\tau$ in direction of $X$ to $Y$} $(X\to Y)$ as
\begin{equation}\label{TailKT}
 \tau_{XY}(q) = 2\mathbb P( (X- \tilde X)(Y - \tilde Y) > 0 \mid X, \tilde X > F_X^{-1}(q)) -1
\end{equation}
where $q\in[0,1]$ is a probability level. Note that this coefficient is asymmetric because of the conditioning that only concerns the first variable. We further note that we only retain observations in the tail of the first variable, that is, the part that exceeds the $q$-quantile $F_X^{-1}(q)$. In extreme value theory, it is common to consider the limit as $q \to 1$ to describe the tail behaviour of the largest observations. We therefore define the limiting \emph{asymmetric tail Kendall's $\tau$ in direction of $X$ to $Y$} as
\begin{equation}\label{lim_kt}
 \tau_{XY} = \lim_{q\to 1} \tau_{XY}(q).
\end{equation}
By exchanging the roles of $X$ and $Y$, we define analogously the asymmetric tail Kendall's $\tau_{YX}(q)$ in direction of $Y$ to $X$ at level $q$ and its limit version $\tau_{YX}$.

Suppose now that we have independent observations $(X_1,Y_1), \dots, (X_n,Y_n) $ of the random vector $(X,Y)$, and let $X_{(1)} \leq X_{(2)} \leq \dots \leq X_{(n)}$ denote the order statistics of the sample $X_i$. Based on the definition in~\eqref{TailKT}, we can define an empirical estimator of the tail Kendall's $\tau$ at level $q =1 - k/n$ as
\[ \widehat \tau_{XY}(q) = \frac{1}{ {k \choose 2}}  \sum_{1\leq i \leq j \leq n} \text{sgn}(X_i - X_j) \text{sgn}(Y_i - Y_j) \einsfun\{X_i, X_j >  X_{(n-k)} \}, \]
where $k \leq n$ is the number of exceedances that are used in the tail, and $\text{sgn}(x)\in \{-1, +1\}$ is the sign of a real number $x \in \mathbb R$. The indicator $\einsfun\{X_i, X_j >  X_{(n-k)} \}$ makes sure that only the $k$ largest observations of $X$ are used. The estimator $\widehat \tau_{YX}(q)$ is defined similarly.

Even when we are interested in the limiting asymmetric tail Kendall's $\tau_{XY}$, in practice, we always have to fix a threshold $q < 1$ and then approximate $\tau_{XY} \approx\widehat \tau_{XY}(q)$. Larger values of $q$ make the approximation $\tau_{XY} \approx \tau_{XY}(q)$ better, but the estimation $\widehat \tau_{XY}(q)$ more unstable, and {vice versa} for lower values of $q$. The choice of this threshold is a therefore a bias-variance trade-off and a long-standing problem in extreme value theory. It is usually resolved by an ad-hoc choice or based on sensitivity analyses or stability plots. The threshold $q$ can also be chosen by the practitioner depending on what events are of particular interest in the application.

\subsection{Properties}

We first note that the asymmetric tail Kendall's $\tau$, both the pre-limit and the limit version, does not depend on the marginal distributions. Similar to the original Kendall's $\tau$, it therefore is only a function of the copula of $(X,Y)$. 
By construction, the asymmetric tail Kendall's $\tau$ also ranges in the interval $[-1 , +1]$.
Positive values mean positive association between extremes, while negative values correspond to negative association. Each single coefficient $\tau_{XY}$ and $\tau_{YX}$ thus tells us something about the strength of association between of the extremes of $(X,Y)$, where the focus is either on those observations where the $X$ or the $Y$ component is large, respectively. Looking at the definition in~\eqref{TailKT}, we see that non-extreme observations in the bulk of the distribution do not impact the asymmetric tail Kendall's $\tau$.
In addition, considering the pair of $(\tau_{XY}, \tau_{YX})$ reveals information on the asymmetry in the extremal structure. If $(X,Y)$ is symmetric, then  $\tau_{XY} = \tau_{YX}$. Otherwise, the difference between the two coefficients is a measure of asymmetry.

A natural question is whether we can compute these coefficients for existing copula models. 
It turns out that this is in fact possible under certain conditions. In~\ref{theory}, we show that the asymmetric tail Kendall's $\tau_{XY}$ can be expressed only in terms of the extreme value copula that corresponds to $(X,Y)$; for background on extreme value copulas we refer to \cite{salvadori2007} and \cite{seg2010}. Formula~\eqref{tau_extreme} in Theorem~\ref{thm_tau_extreme} (in~\ref{theory}) gives an explicit expression in terms of the so-called extremal functions. Since for most extreme value copulas the distribution of the extremal function is known, it is very easy to compute $\tau_{XY}$ numerically through Monte Carlo methods. We provide here three examples.

\begin{example}
If $X$ and $Y$ are independent, then $\tau_{XY}=\tau_{YX}=0$.
\end{example}

\begin{example}\label{ex:HR}
    The class of bivariate H\"usler--Reiss distributions is a parametric family parameterized by $\Gamma \in (0,\infty)$ \citep{HR1989}.  If $(X,Y)$ follows a H\"usler--Reiss distribution, it is well know that the extremal correlation in~\eqref{chi} is given by $\chi_{XY} = 2\{1 - \Phi(\sqrt{\Gamma} / 2)\}$, where $\Phi$ is the distribution function of a standard normal distribution. It can be checked that for $\Gamma \to 0$ the dependence structure approaches complete dependence, and for $\Gamma \to \infty$ independence.
    For this distribution class, we can also compute the asymmetric tail Kendall's $\tau$ in closed form as
    \[\tau_{XY} = \tau_{YX} = 2\exp(\Gamma) \{1 - \Phi(\sqrt{2\Gamma})\}; \]
    seein~\ref{theory} for the details.
    Since H\"usler--Reiss distributions are symmetric, we see that both asymmetric tail Kendall's $\tau$ coefficients are equal. The left panel of Figure~\ref{fig:tau_chi} shows both the extremal correlation and the asymmetric tail Kendall's $\tau$ as a function of the parameter $\Gamma$.
\end{example}

\begin{example}\label{ex:dirichlet}
    The bivariate extremal Dirichlet model distribution is a parametric family with two parameters $\alpha_1, \alpha_2 \in (0,\infty)$. If $(X,Y)$ follows this distribution, neither the extremal correlation nor the asymmetric tail Kendall's $\tau$ can be computed in closed form. Using formula~\eqref{tau_extreme}, the latter can numerically be approximated based on the results on the corresponding extremal functions; see for instance \citet[][Example 3]{eng2020}.
    Extremal Dirichlet models are symmetric only if $\alpha_1 = \alpha_2$, otherwise these two parameters govern the amount of asymmetry. The right panel of Figure~\ref{fig:tau_chi} shows both the extremal correlation and the asymmetric tail Kendall's $\tau$ as a function of the parameter $\alpha_2$ for a fixed value of $\alpha_1 = 2$. We see that in this case the two curves for the asymmetric tail Kendall's $\tau_{XY}$ and $\tau_{YX}$ are different, and only coincide in the symmetric case where $\alpha_1 = 2$.
    Changing $\alpha_1$ results in a similar picture but the overall strength of dependence would shift up or down. 
\end{example}

\begin{figure}[h]
    \centering
    \includegraphics[width=0.49\textwidth]{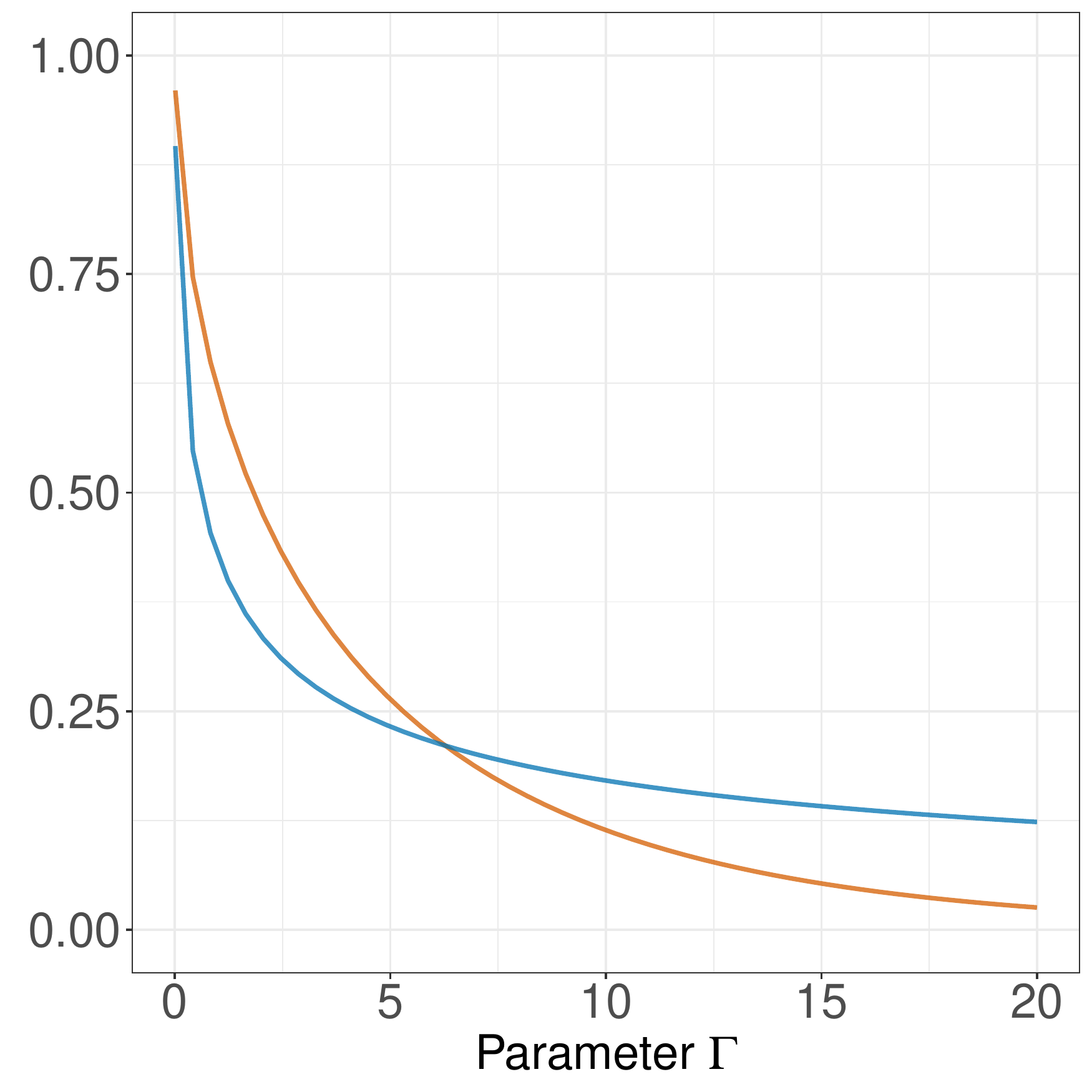}
    \includegraphics[width=0.49\textwidth]{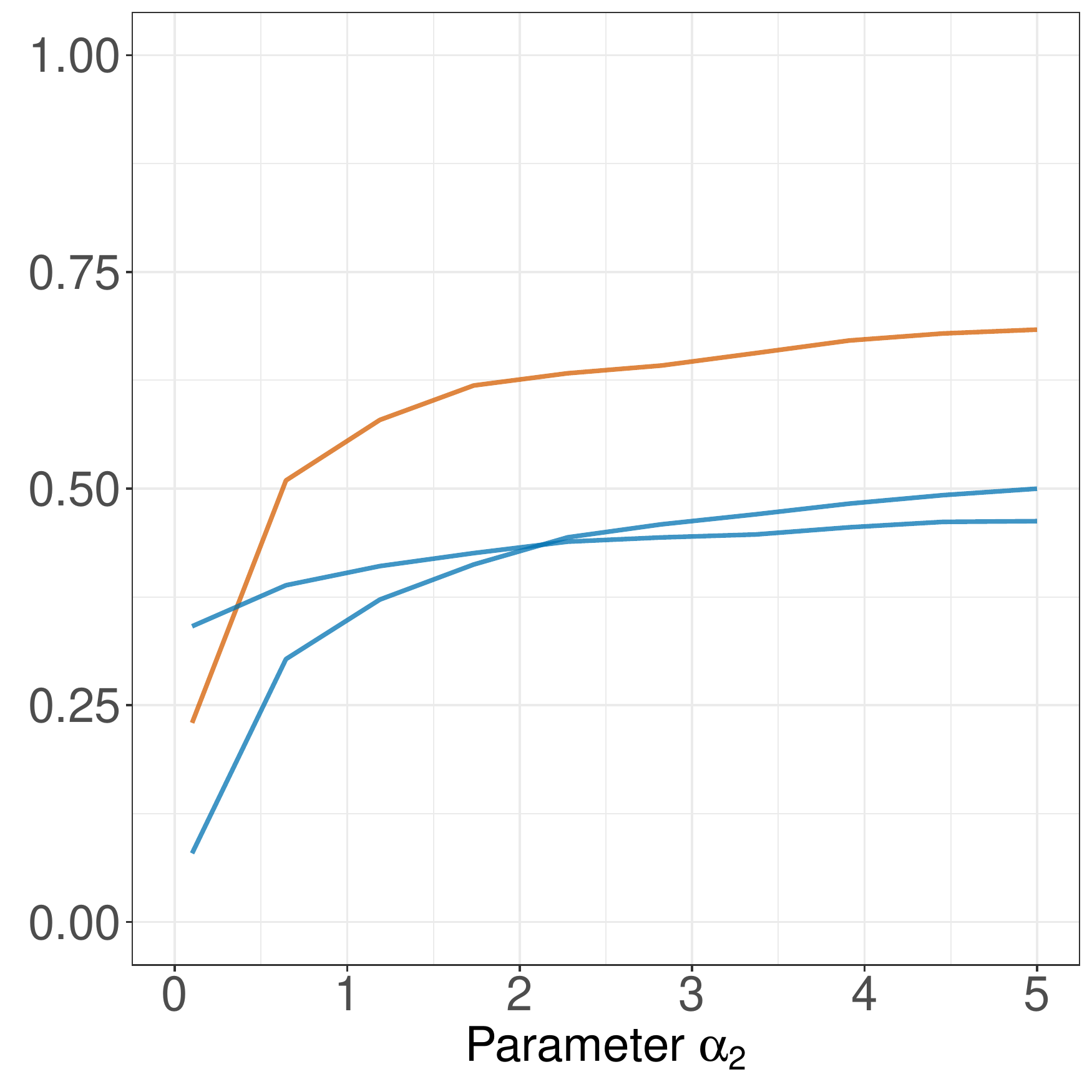}
    \caption{Left: extremal correlation (orange) and asymmetric tail Kendall's $\tau$ (blue) for the H\"usler--Reiss distribution as function of $\Gamma$. Right: same for the Dirichlet family with parameter $\alpha_1=2$ fixed and as function of the second parameter $\alpha_2$.}
    \label{fig:tau_chi}
\end{figure}

\subsection{Numerical simulations}

In hydrological applications, when considering upstream and downstream stations on the same river, there is a natural direction or asymmetry given by the water flow of the river. Due to this fact, many asymmetric copula models have been proposed. This includes the asymmetric versions of the Gumbel--Hougaard, Galambos and H\"usler--Reiss copulas obtained by the asymmetrization technique in \cite{Genest11}.
Our asymmetric tail Kendall's $\tau$ provides us with a method to detect and quantify possible asymmetries. Before applying it to a real-world data set, we conduct several numerical experiments to assess the performance of our approach on simulated data where we know the true dependence.

In the simulations we use the asymmetric logistic copula defined as
\begin{equation}
    C_{\alpha,\beta_1,\beta_2} (u,v) = \exp\left\{-\left[(-\beta_1 \log u)^\alpha + (-\beta_2 \log v)^\alpha\right]^{1/\alpha} + (1- \beta_1) \log u + (1-\beta_2) \log v) \right\}
\end{equation}
where $\beta_1$, $\beta_2$ $\in$ [0,1], $\alpha \geq 1$. The parameters $\beta_1$ and $\beta_2$ determine the asymmetry of the copula and $\alpha$ is specifies the overall strength of dependence \citep{Gagliardini22}. If $\beta_1 = \beta_2 = 0$ then we obtain the independence copula. On the other hand, if $\beta_1 = \beta_2 = 1$, the we recover the symmetric logistic copula.

To perform the simulations we use the function \texttt{remved} in the  \texttt{evd} R package \citep{Evdpackage}. We consider asymmetric logistic copulas with different combinations of parameters. The asymmetry parameters $\beta_1$ and $\beta_2$ range from $0.1$ to $0.9$ with step size $0.1$, and we choose $1/\alpha\in \{0.005,0.2,0.4,0.6,0.8,0.98\}$.
For each of the 486 combinations of the three parameters we simulate $n=1000$ random samples of $(X,Y)$ from the corresponding asymmetric logistic copula. 
For each set of samples we calculate both asymmetric tail Kendall's $\widehat{\tau}_{XY}(q)$ and $\widehat{\tau}_{YX}(q)$, where we fix the probability level to $q =0.98$, which corresponds to $k=20$ exceedances; the results are fairly stable for different values of $q$. We further compute the absolute difference $|\widehat{\tau}_{XY}(q) - \widehat{\tau}_{YX}(q)|$ as a measure of asymmetry and the maximal value $\max\{\widehat{\tau}_{XY}(q), \widehat{\tau}_{YX}(q)\}$ as a measure of overall dependence.

\begin{figure}[tb!]
    \centering
    \includegraphics[width=\textwidth]{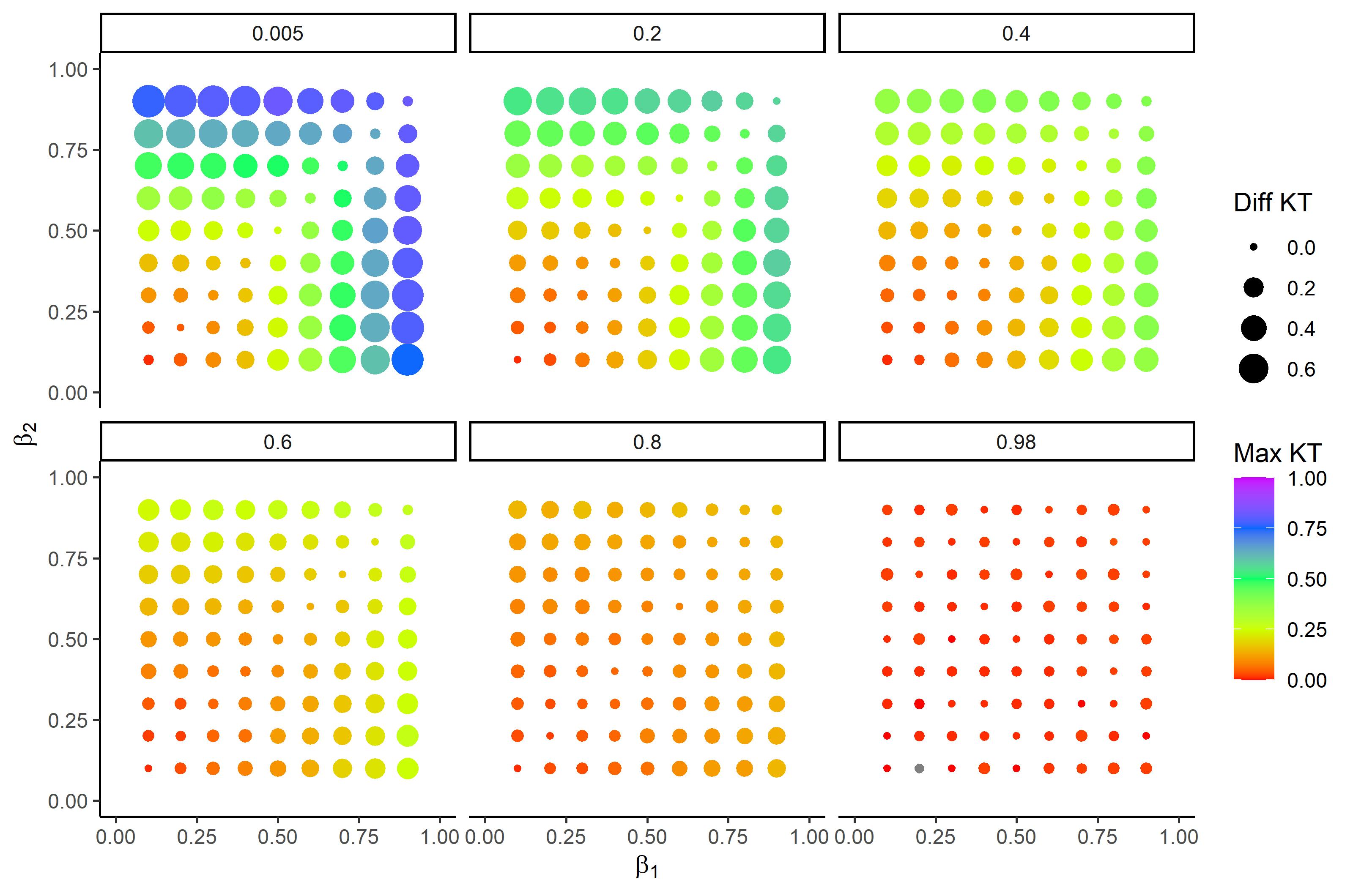}
    \caption{Simulations of sample size with asymmetric logistic copula. Each point is the median value of asymmetric tail Kendall's $\tau$ considering 1000 samples for each set of combinations of Copula parameters: $\beta_1$ on the $x$-axis, $\beta_2$ on the $y$-axis , $1/\alpha$ varying for each panel. The size of the point is the difference of asymmetric tail Kendall's $\tau$ in the two directions for each point and the colour is the maximum value of  asymmetric tail Kendall's $\tau$ for each pairwise. }
    \label{fig:SimAsyCop}
\end{figure}

The results are summarized in Figure~\ref{fig:SimAsyCop}. On the $x$-and $y$-axes there are the values of $\beta_1$ and $\beta_2$, respectively. Each panel shows simulations for a particular value of $1/\alpha$. The size of each point represents the difference of the two asymmetric tail Kendall's $\tau$ values and the colour is the maximum asymmetric tail Kendall's $\tau$ among the two directions. Each point is the median value over 1000 simulations.
When the the dependence parameter $1/\alpha$ is close to one (last panel on the right in Figure~\ref{fig:SimAsyCop}), we have approximately independent samples, which is reflected by coefficients close to zero in both directions. 
Decreasing the value of $1/\alpha$ increases the dependence, which is generally larger asymmetric tail Kendall's $\tau$ values. In these cases, we can appreciate better the asymmetry in the copula given by the different values of $\beta_1$ and $\beta_2$. Looking at the asymmetric tail Kendall's $\tau$ values, we observe that the larger the difference in the these parameters, the larger the difference between the asymmetric tail Kendall's $\tau$ coefficients. Along the diagonals we have $\beta_1=\beta_2$ so we are effectively sampling from symmetric but dependent distributions; the differences of the asymmetric tail Kendall's $\tau$ values are consequently very low there. This confirms that our new pair of coefficients is well suited for detecting and quantifying asymmetry in data. 

To understand why there can be a directionality in the data that leads to asymmetry in the extremal dependence, we recall that a random vector with standard Fr\'echet margins and asymmetric logistic copula $C_{\alpha, \beta_1, \beta_2}$ can be represented as 
\begin{align}\label{al_rep}
 (X,Y)=\left(\max\{\beta_1 V, (1-\beta_1) \varepsilon_1  \}, \max\{\beta_2 W, (1-\beta_2)\varepsilon_2 \} \right),
 \end{align} 
where $\varepsilon_1, \varepsilon_2$ are independent standard Fr\'echet variables and $(V,W)$ follow a symmetric logistic distribution with parameter $\alpha$. We can therefore see an asymmetric logistic distribution as a max-mixture of a symmetric logistic distribution and independent noise. For instance, suppose that $X$ and $Y$ represent the discharges at a downstream and upstream station on the same river, respectively. Then a possible model for the extremal dependence would be~\eqref{al_rep} with $\beta_2 = 1$ since the largest discharges at $X$ would come either from station $Y$ (represented by the strongly dependent $(V,W)$) or from runoff coming from later tributaries independent of $Y$ (represented by $\varepsilon_1$). In this case one can check that $\tau_{XY} < \tau_{YX}$. 
This provides the first intuition why in hydrological data we expect to see asymmetry in the dependence between the largest observations, and more specifically, a directionality from upstream to downstream. 

In order to show that we can capture this directionality with our asymmetric tail Kendall's $\tau$ coefficient, in Figure~\ref{fig:Directional_plot_sim} we plot the points from Figure~\ref{fig:SimAsyCop} according to their asymmetric tail Kendall's $\tau$ values. On the $x$-axis, we always put the asymmetric tail Kendall's $\tau_{XY}$ where $X$ represents the "downstream" station, that is, the variable with the smaller $\beta$ coefficient. We observe that essentially all points are above the diagonal, which indicates that the directionality in the data can be identified from the pair of asymmetric tail Kendall's $\tau$ coefficients: if $\tau_{XY} < \tau_{YX}$ then $X$ is downstream of $Y$. The color coding in the figure shows the difference of $\beta$ coefficients, a proxy for the strength of asymmetry. In general, the larger the difference the more clearly asymmetry is detected.

\begin{figure}[tb!]
    \centering
    \includegraphics[width=\textwidth]{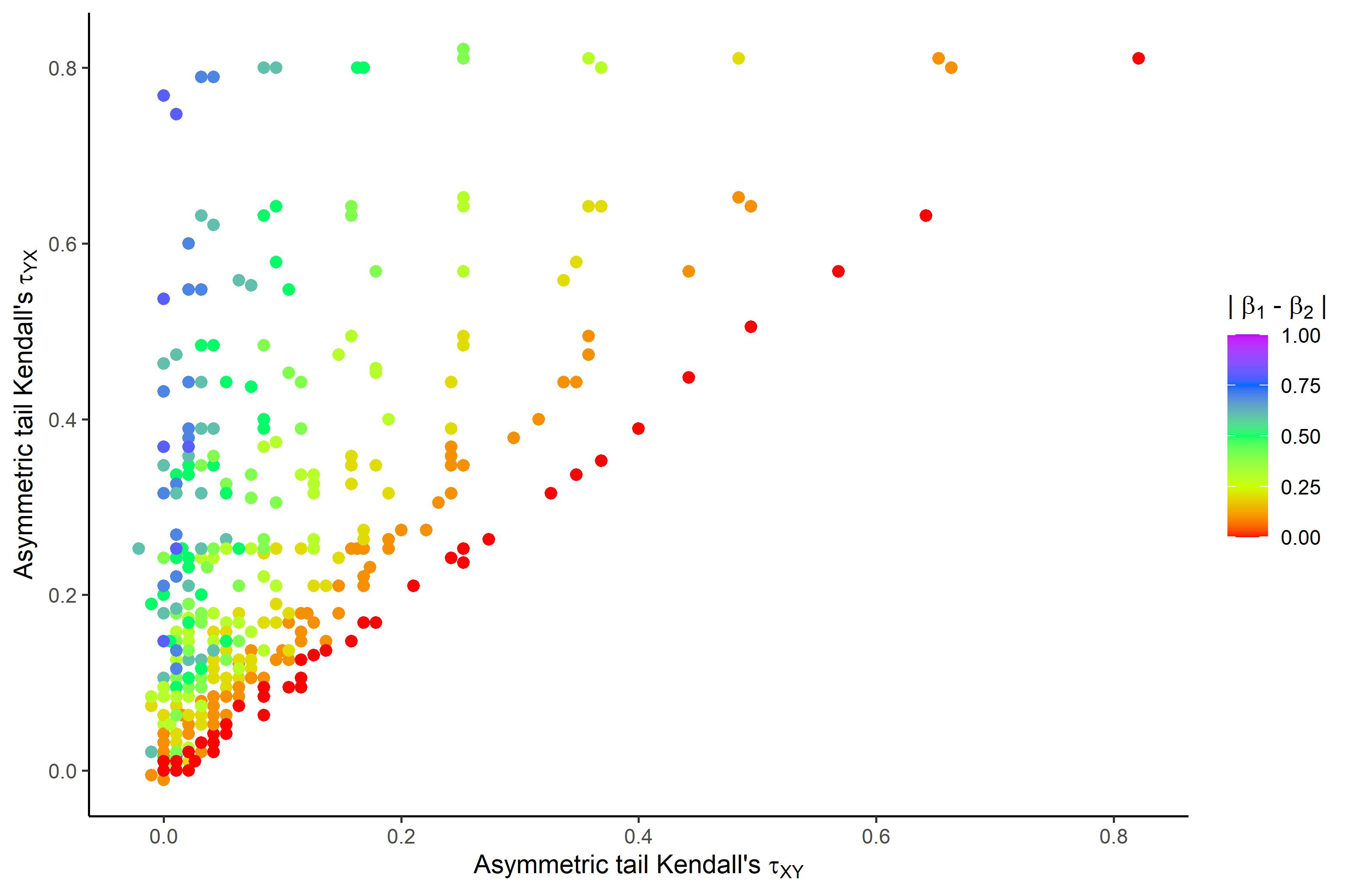}
    \caption{Plot of asymmetric tail Kendall's $\tau$ coefficients for data simulated according to asymmetric logistic copulas; on the $x$-axis is the coefficient that corresponds to conditioning on the "downstream" station, that is, the one with the smaller $\beta$ coefficient. The colors encodes the absolute difference between the $\beta$ coefficients.}
    \label{fig:Directional_plot_sim}
\end{figure}

\subsection{Causality}\label{Causality}
Understanding causal relationships in data is of utmost importance for hydrology, climate research and risk assessment \citep{Ombadi2020}.
Different methodologies have been proposed to study causality for environmental and climate data \citep{Runge2019,DiCapua20,Immecausal}. 
There is a growing literature on methods that focus on detecting causal structures between extreme observations \citep{Mhalla2019, gnecco2021causal, tra2021} or on estimating treatment effects on extreme outcomes \citep{Deuber}.
In this section we discuss that our tail coefficient has a causal interpretation in certain model classes.

One of the most important models in causality is given by linear structural equation models. For two variables $X$ and $Y$, where $X$ causes $Y$, they can be written as 
\begin{align*}
    X &= \varepsilon_1 \\
    Y &= \beta X + \varepsilon_2,
\end{align*}
where $\varepsilon_1, \varepsilon_2$ are independent noise variables.
If the noise variables are non-Gaussian, it has been argued that the asymmetry in the distribution of $(X,Y)$ can be used to detect causal directions \citep{shimizu2006linear}. For heavy-tailed noise, this has been shown to be true also for extremal dependence in general, and for applications in hydrology in particular \citep{gnecco2021causal}.

We perform a simulation study to assess whether our asymmetric tail Kendall's $\tau$ is also able to capture such causal directions in linear structural causal models. We generate $n=4000$ samples of $(X,Y)$ in the above linear structural equation model where $\varepsilon_1$ and $\varepsilon_2$ have Student-t distributions with 3 degrees of freedom. 
Between two variables there are three different configurations (see top row of Figure~\ref{fig:Sixcases}): either $X$ and $Y$ are independent ($\beta = 0$), $X$ causes $Y$ ($\beta\neq 0$) or $Y$ causes $X$ ($\beta \neq 0$ and exchange roles of $X$ and $Y$). In our simulation we choose $\beta = 0.3$ for the latter two situations.
Figure~\ref{fig:SimCausality} shows the boxplots of the asymmetric tail Kendall's $\tau$ estimates with a probability threshold of $q = 0.98$ ($k=80$) based on 1000 repetitions of these simulations. We show how the two asymmetric tail Kendall's $\tau$ coefficients clearly indicate the direction of causality: if $X$ and $Y$ are independent then $\widehat \tau_{XY}(q) \approx \widehat \tau_{YX}(q)$ if $X$ causes $Y$ then $\widehat \tau_{XY}(q) > \widehat \tau_{YX}(q)$, and vice versa if $Y$ causes $X$. 
As comparison we also show the boxplots of estimates of a symmetric version of asymmetric tail Kendall's $\tau$, which is the classical Kendall's $\tau$ applied to the data where both variables exceed the $q$-quantile. We see that this coefficient measures the overall dependence in the data, but cannot be used to detect asymmetries or causal links.

In causality, it is possible that the two variables $X$ and $Y$ are further influenced by a possibly unobserved confounder (e.g., precipitation for river discharge); see bottom row of Figure~\ref{fig:Sixcases}. In this case it becomes more challenging to detect whether there is a causal link between the variables, and in which direction it points. We run the simulations again but now with an unobserved confounder $C$. The bottom row of Figure~\ref{fig:SimCausality} shows that then the difference of the asymmetric tail Kendall's $\tau$ coefficients becomes less pronounced when there is a causal link. While causal detection becomes more difficult, it can be seen that it is still possible in this framework.

\begin{figure}[h]
    \centering
    \includegraphics[width=\textwidth]{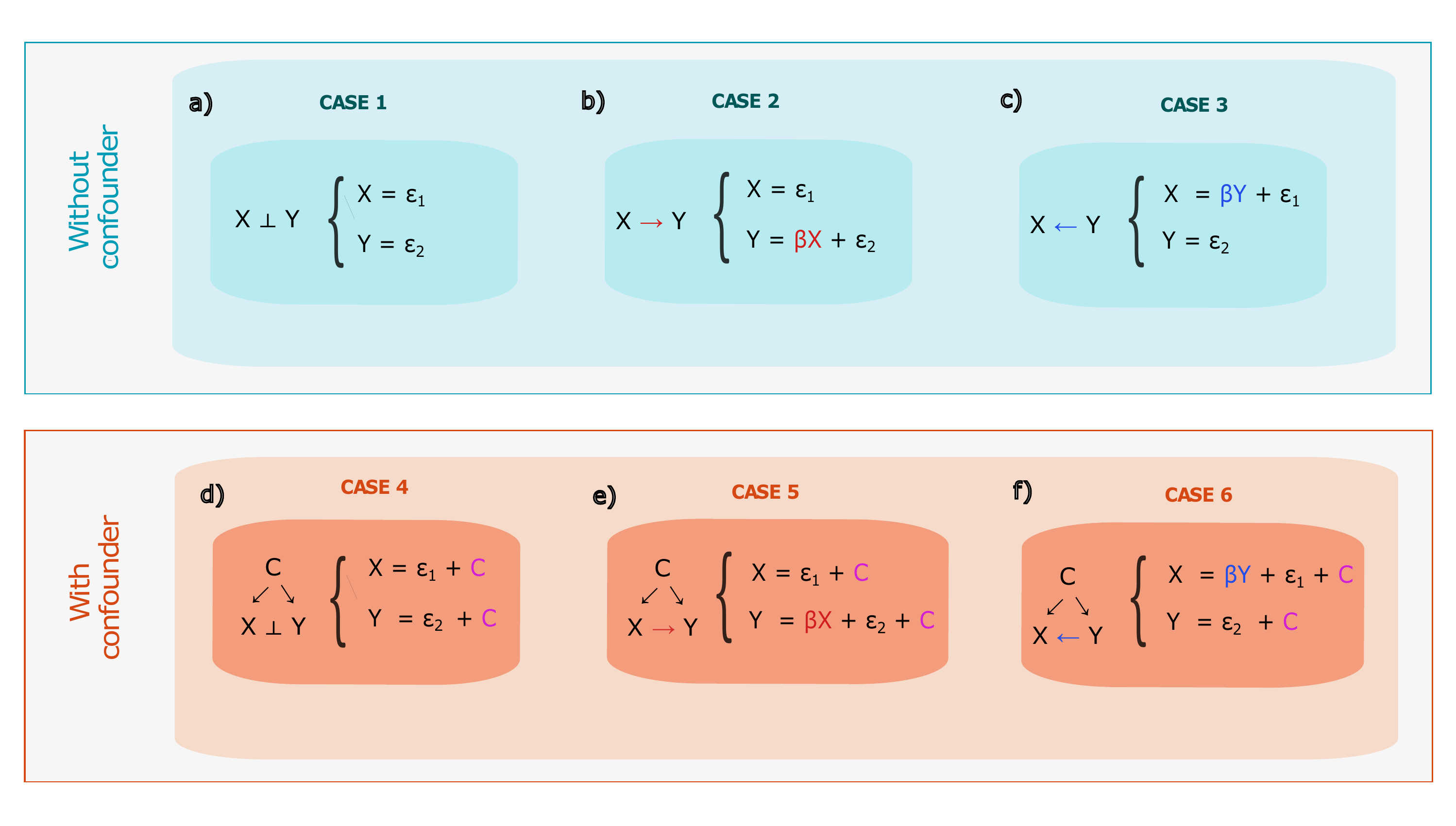}
    \caption{Six possible causal models of two variables $X$ and $Y$ without (top row) and with (bottom row) confounder $C$. In panels a) and d) $X$ and $Y$ are independent, in panels b) and e) X causes $Y$, and in panels c) and f) Y causes $X$.}
    \label{fig:Sixcases}
\end{figure}

\begin{figure}[h]
    \centering
    \includegraphics[width=\textwidth]{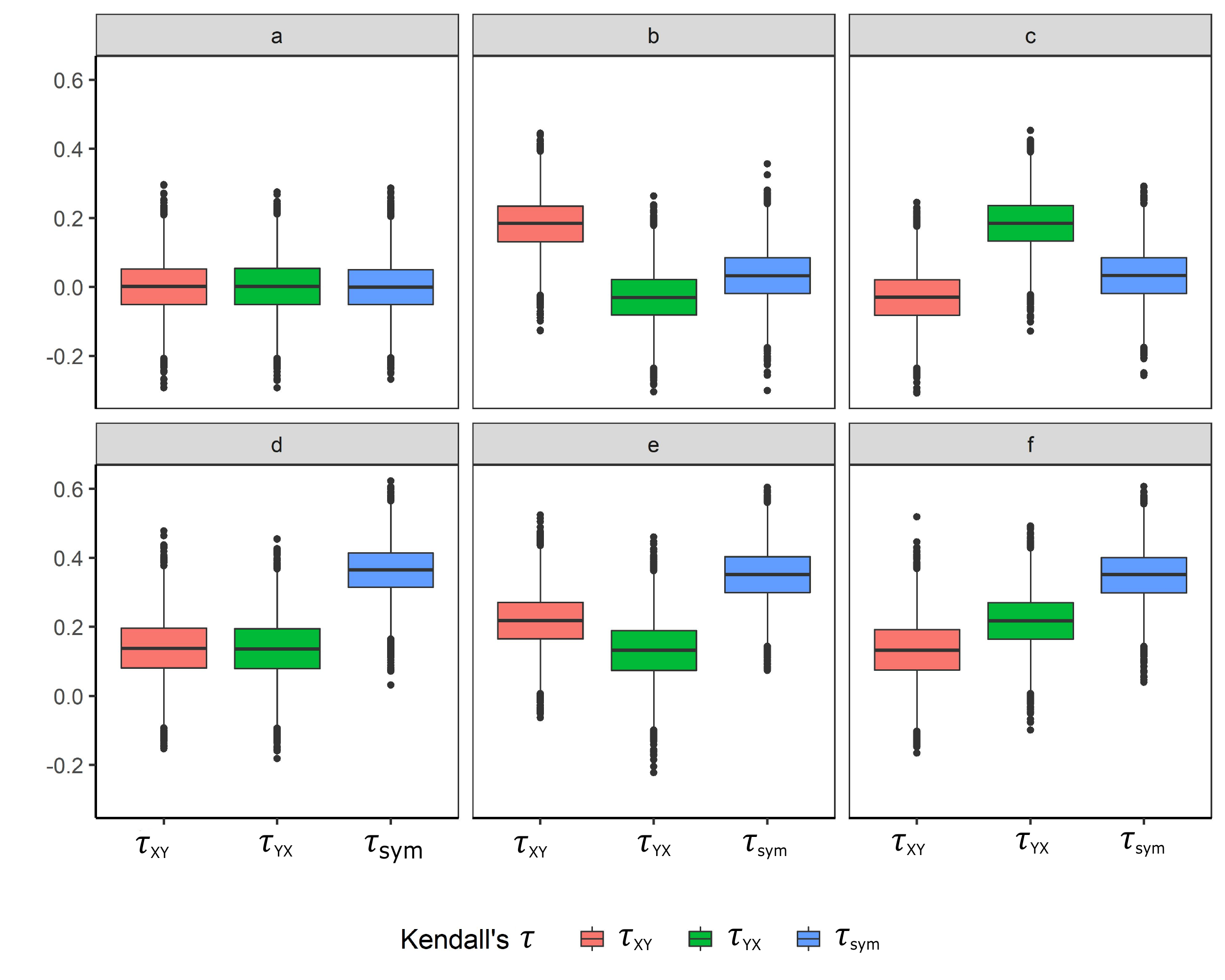}
    \caption{Boxplots of estimates of $\tau_{XY}$, $\tau_{YX}$ and the symmetric asymmetric tail Kendall's $\tau$ ($\tau_{\text{sym}}$) for 1000 simulations of the six possible causal models.}  
    \label{fig:SimCausality}
\end{figure}

\section{Hydrological application}

\subsection{Case study and data analysis}

\begin{figure}[htb!]
    \centering
    \includegraphics[width=0.6\textwidth]{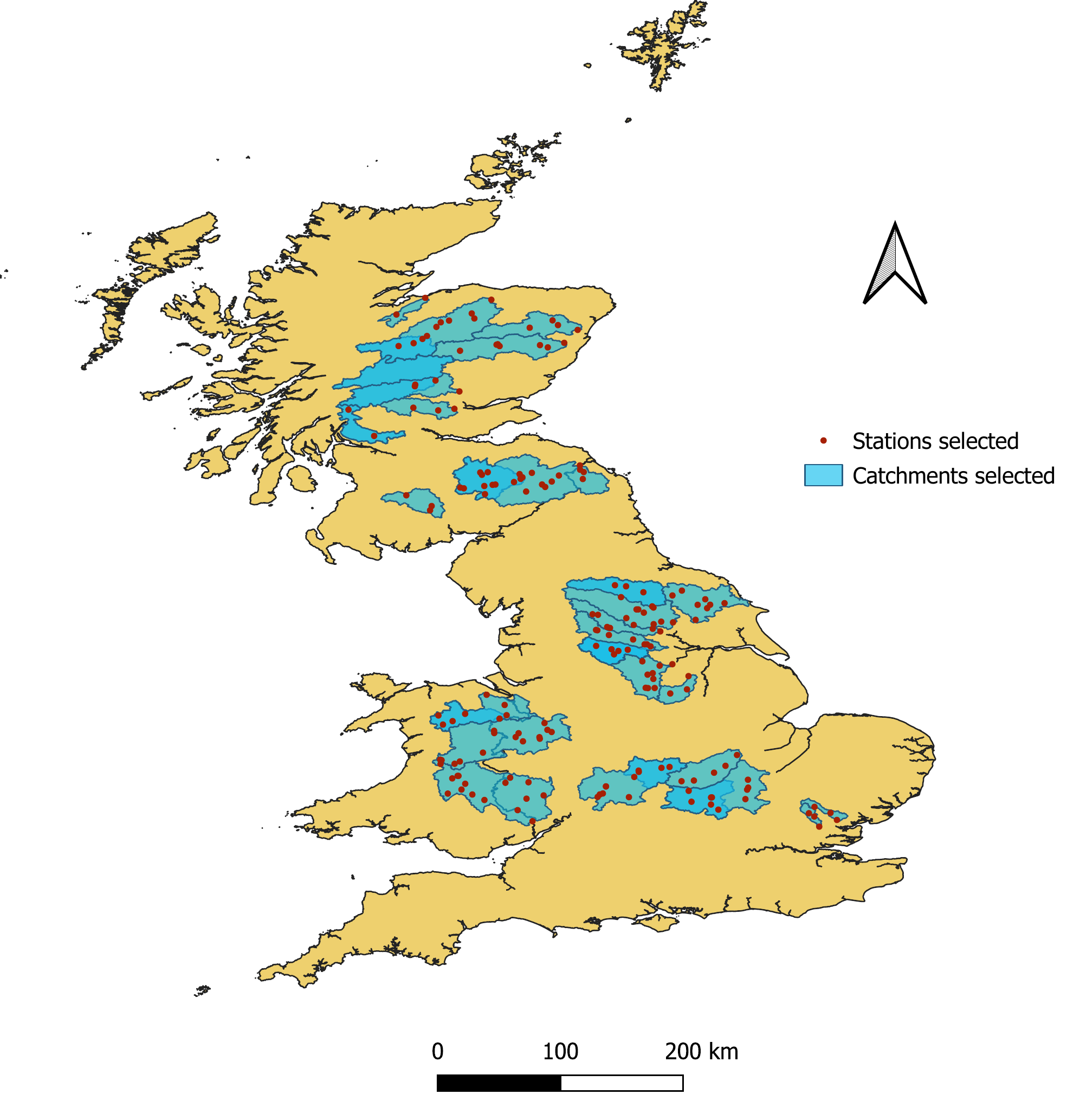}
    \caption{Stations (178) and catchments (18) selected in the United Kingdom.}
    \label{fig: map_Uk}
\end{figure}

\begin{figure}[htb!]
    \centering
    \includegraphics[width=\textwidth]{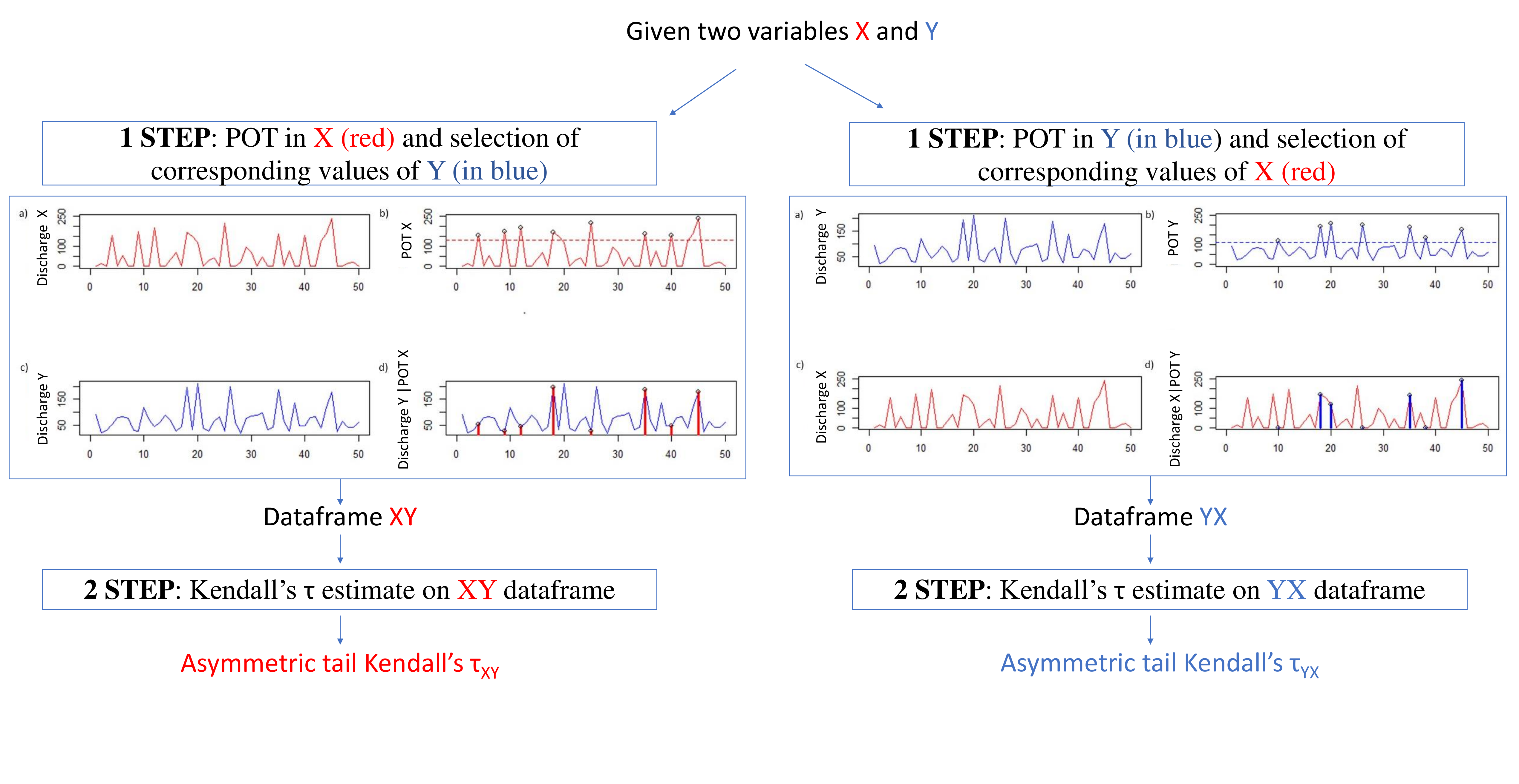}
    \caption{Procedure for extracting the observations and computing the asymmetric tail Kendall's $\tau$ in both directions:  in left panel $\tau_{XY}$, and in right panel $\tau_{YX}$.}
    \label{fig: KTscheme}
\end{figure}

Extreme value analysis in river networks is of primary importance for hydrological applications \citep[e.g.,][]{asa2018,eng2018}. There is a natural directionality that arises from the fact that some stations are located upstream or downstream  relative to other stations. This directionality is likely to induce an asymmetry in the dependence structure between the largest observations. In this section we apply our asymmetric tail Kendall's $\tau$ defined in Section~\ref{sec:def} to daily discharge data on 18 basins in the United Kingdom. The data set is obtained from the National River Flow Archive (available at \texttt{https://nrfa.ceh.ac.uk/}) for the United Kingdom. It has been used in many other works including \citep{deidda2021causes,Fry2013,Robson1999,Dixon2013,Chang2003a}. We focus on 18 catchments for which there are at least 2 stations in the main river channel and where daily data are available. In total there are 178 gauging stations.

For our data analysis we consider all the possible pairs of stations resulting in $15,753$ combinations. For any pair of stations $X$ and $Y$, we select all years where both stations have observed data. We then estimate the two asymmetric tail Kendall's $\tau$ coefficients $\widehat \tau_{XY}(q)$ and $\widehat \tau_{YX}(q)$ at probability level $q = 0.98$ following the scheme reported in Figure~\ref{fig: KTscheme}.
From the intuition discussed after~\eqref{al_rep}, we expect to see asymmetry in the extremal dependence whenever $X$ is upstream or downstream of $Y$. In particular if $X$ is downstream of $Y$, we conjecture that $\tau_{XY} < \tau_{YX}$. If there is no such connection, either because the stations are on different tributaries or even in different basins, we would expect symmetry between the two stations and approximate equality between the two coefficients. In order to check these conjecture on our data, we extract from a Geographic Information Systems (GIS) the sub-catchments for each stations and for each pair determine which relation they have.

This analysis is a check whether our asymmetric tail Kendall's $\tau$ also shows the desired properties on real data that have been observed in the previous sections in simulations. Beyond this, such an exploratory analysis is crucial to determine the type of statistical model that can be used in subsequent steps. Our new coefficients allows to assess not only the strength of extremal dependence but also a possible asymmetry. 

\subsection{Results}

\begin{figure}[tb!]
    \centering
    \includegraphics[width=\textwidth]{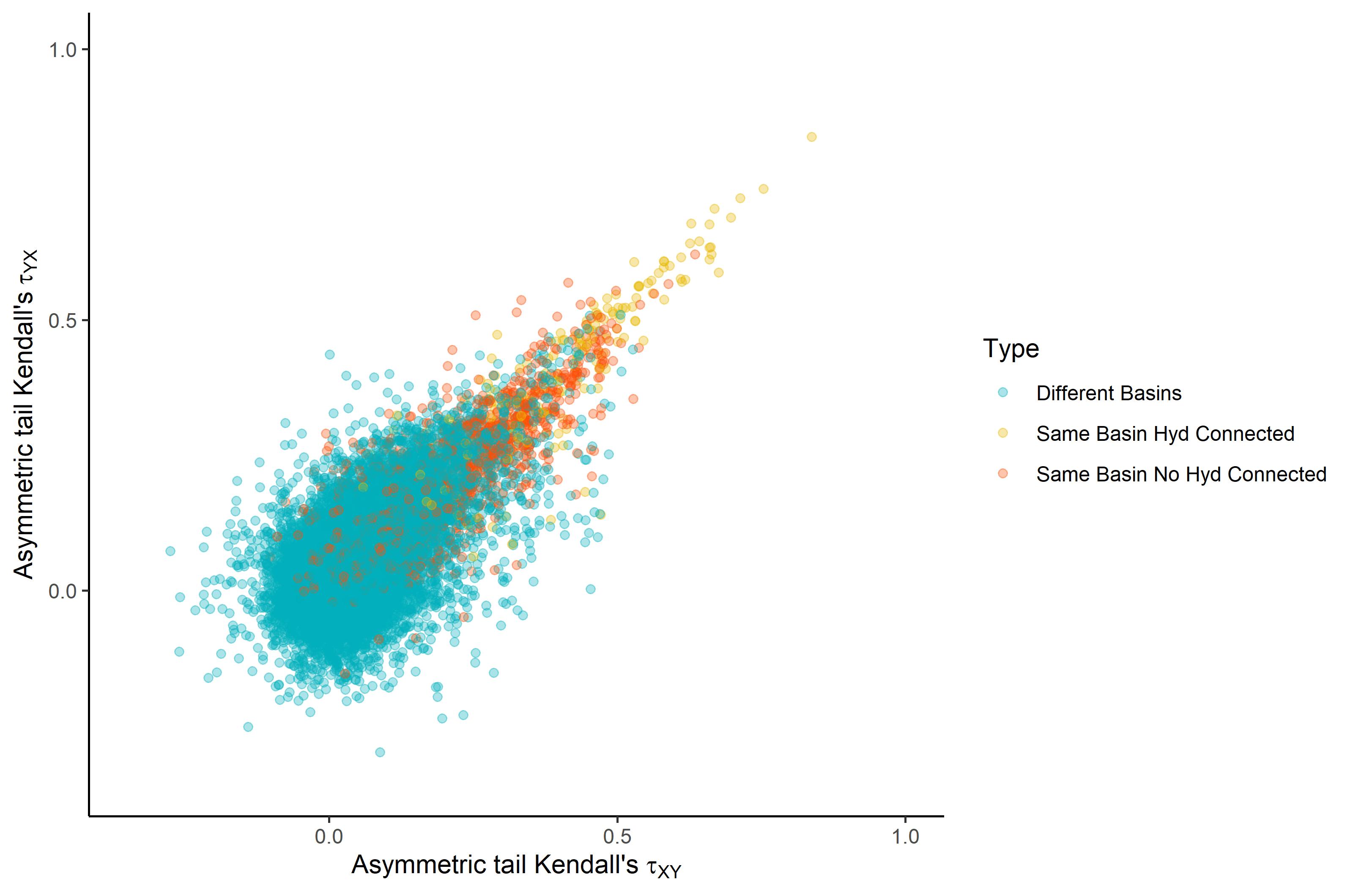}
    \caption{Asymmetric tail Kendall's $\tau$ in both direction considering $15,753$ couples. Blue, red and yellow points are couples from different basins, hydraulically connected couples in the same basin and hydraulically unconnected in the same basin, respectively. }
    \label{fig: KT_group}
\end{figure}

The estimates of all asymmetric tail Kendall's $\tau$ coefficients for both directions are shown in  Figure~\ref{fig: KT_group}. The points are grouped with different colours according to whether they are in different basins or in the same basin, and in the latter case whether they are hydraulically connected or not. 
We observe that the stations that are in different basins concentrate in the lower bottom corner, indicating that their association is the weakest. This is to be expected since they are not hydraulically connected and might be impacted from fairly different precipitation events.
Stations that are in the same river basin tends to have stronger extremal dependence, shown by larger asymmetric tail Kendall's $\tau$ values. Unsurprisingly, those pairs that are hydraulically connected exhibit the strongest association.
Note that this plot is symmetric since we did not order the coefficients from upstream to downstream as in Figure~\ref{fig:Directional_plot_sim}. The reason is that most stations are not hydraulically connected and such an ordering would therefore not make sense.

To better explore the asymmetric tail behaviour of hydraulically connected couples, similarly to Figure~\ref{fig:Directional_plot_sim}, we plot in Figure~\ref{fig:DirectionalKT} the corresponding asymmetric tail Kendall's $\tau$ coefficients ordered from upstream to downstream on the $x$-axis, and \emph{vice versa} on the $y$-axis.
We see a clear asymmetry in the plots: most of the points are above the diagonal meaning that $\tau_{XY} < \tau_{YX}$ if $X$ is downstream of $Y$. This confirms our conjecture based on the theory and the simulations in Section~\ref{sec:methods}.
Comparing to Figure~\ref{fig:Directional_plot_sim}, we observe that in real data the detection of asymmetry can be more subtle than in certain simulations, and therefore it is important to have tailored methods such as our asymmetric tail Kendall's $\tau$.

\begin{figure}[tb!]
    \centering
    \includegraphics[width=0.8\textwidth]{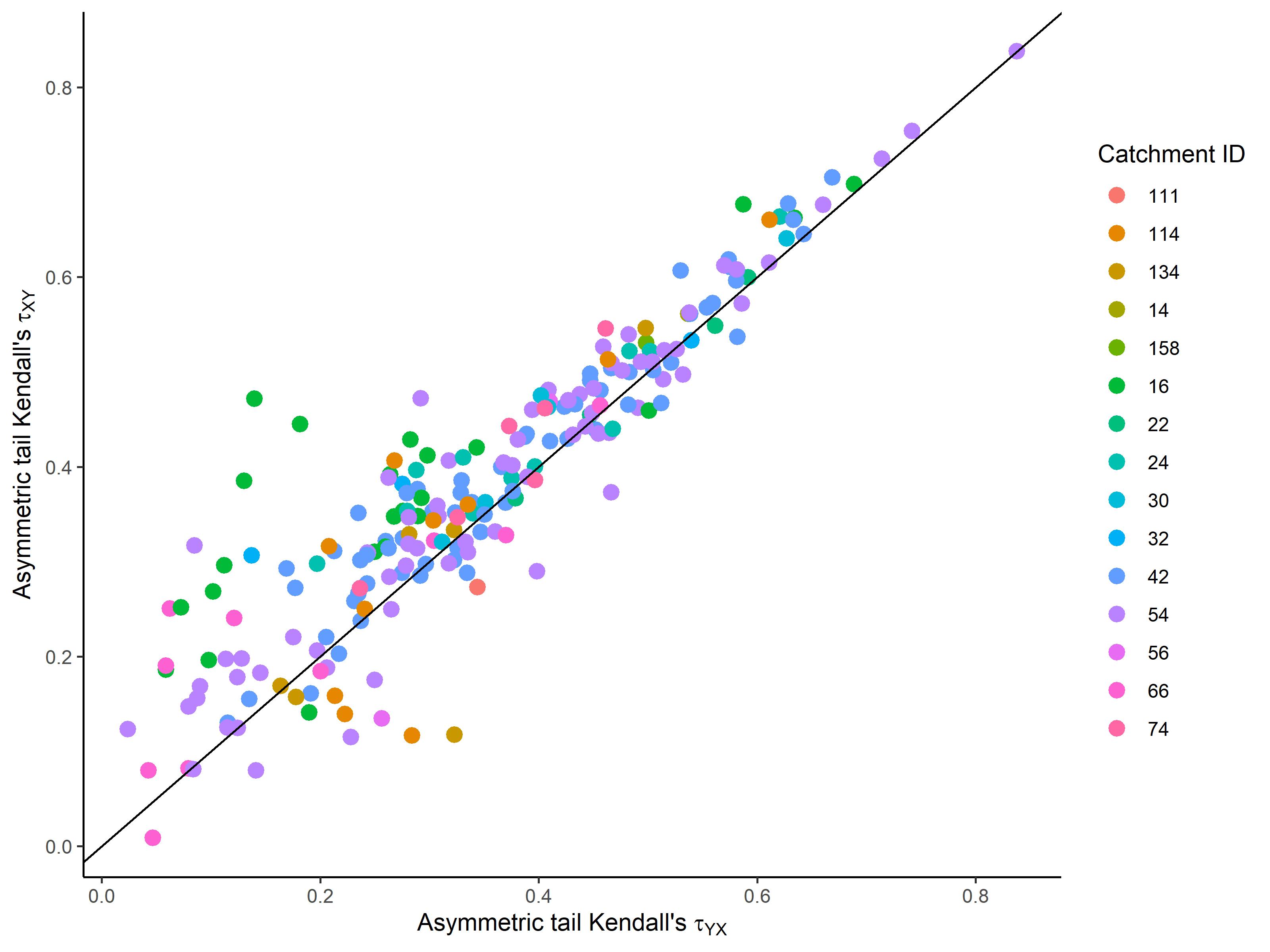}
    \caption{Plot of asymmetric tail Kendall's $\tau$ coefficients for hydraulically connected pairs; on the $x$-axis is given the coefficient that corresponds to conditioning on the ”downstream” station, that is.}
    \label{fig:DirectionalKT}
\end{figure}

While in general, our coefficients agree with the direction of the stream, there are some points that do not and whose points are below the diagonal. A prominent example is a point in the Avon basin (catchment ID = 114), which is far below the diagonal in Figure~\ref{fig:DirectionalKT}.
We look into this basin in more detail in Figure~\ref{fig:map_comparison57}. In panel a) we show the real direction of river flow and in panel b) the results of our analysis, where the arrow is from $Y \to X$ if $\tau_{XY} < \tau_{YX}$.
It can be seen that there are two arrows on the northern part that do  not correspond to the real direction of the flow. In fact, this behaviour can be related to the Stanford Reservoir, which is placed in the upper part of the Avon basin and influences the most upstream station. The presence of dams and reservoirs can influence the peak runoff, which then impacts the observed direction of dependence. 
It is interesting that we can detect such behavior in the data. It is an important indication for statistical modeling of this basin.

The results of Figures~\ref{fig:DirectionalKT} and~\ref{fig:map_comparison57} explain the difference in asymmetric tail Kendall's $\tau$ for hydraulically connected points.
For stations in the same basin that are not hydraulically connected or stations in different basins, in Figure~\ref{fig: KT_group} we may still see difference in these coefficients. The reason for this can be either estimation error, or other meteorological conditions that introduce asymmetry in the river discharges.

\begin{figure}[tb!]
    \centering
    \includegraphics[width=1.1\textwidth]{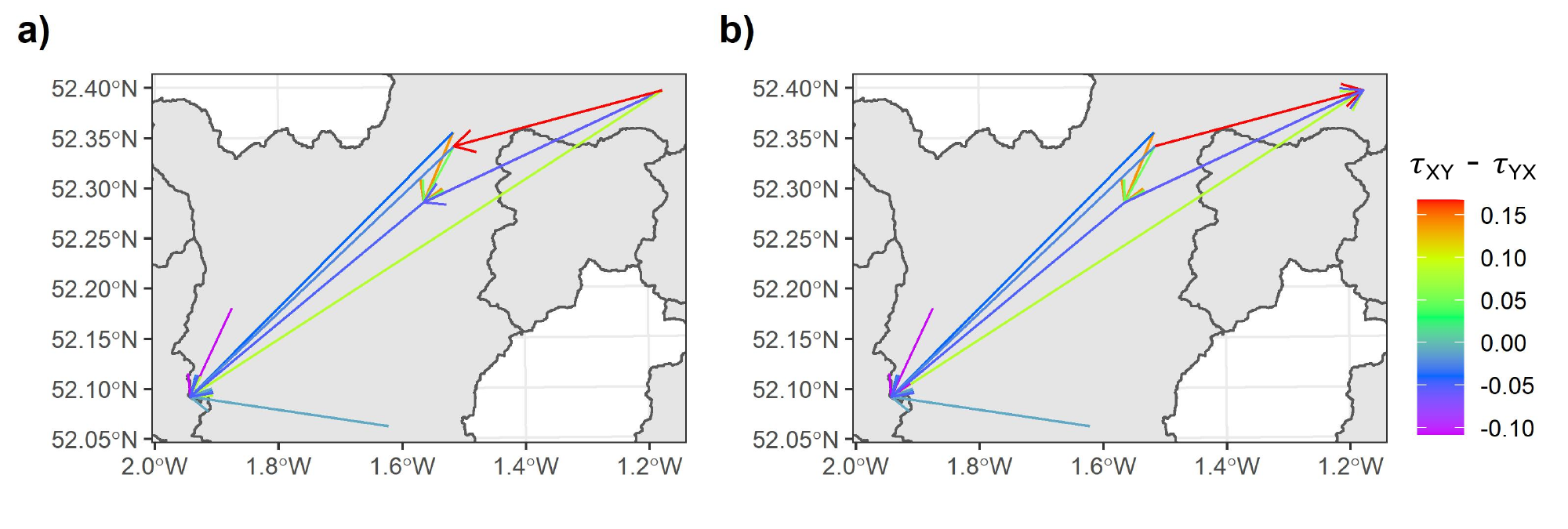}
    \caption{Comparison among real flow directions in panel a) and arrows in panel b) based on asymmetric tail Kendall's $\tau$ estimates for each couple in the Avon catchment.}
    \label{fig:map_comparison57}
\end{figure}

\section{Conclusions}

The study of dependence in extremes is of primary importance in hydrology, climate science and many other fields. It can be used for better prediction of flood events \citep{Olivier2022} or for causal analysis between different variables. 
The common practice for the analysis of extremal dependence is to rely on symmetric summary statistics such as the extremal correlation in~\eqref{chi}. 
In climate science, the study of compound extremes
has attracted increasing attention in recent years. 
Also here, mostly symmetric coefficients are used for model assessment \citep{bou2022} or clustering methods \citep{vig2021, zsc2021}.

In this work, we propose a "conditional" tail version of Kendall's $\tau$ that is able to detect and quantify asymmetric relationship among variables. The asymmetric tail Kendall's $\tau$ is computed for two directions, first conditioning on one variable and then on the other. It only takes into account the largest observations (in hydrology corresponding to flood events) of a data set, and therefore assesses tail risk.
We show in simulations and real data the effectiveness of our coefficient. We also discuss how it can be used to identify causal directions in linear structural equation models. 

In the hydrological application we find that asymmetric tail Kendall's $\tau$ performs well to detect asymmetries between upstream and downstream stations. For unconnected stations, such asymmetry may also be detected. This is statistically relevant but must be explained on a case-to-case basis based on domain knowledge on the meteorological conditions of the catchment.
The practical benefit of this insight is that models can be chosen to incorporate this asymmetry appropriately, for instance through using asymmetric copula models. Also after the model fitting, estimates of asymmetric tail Kendall's $\tau$ give a more refined way of checking the model fit than simpler symmetric coefficients.

\bibliographystyle{Chicago}
\bibliography{Biblio}

\begin{thebibliography}{}

\bibitem[\protect\citeauthoryear{Asadi, Davison, and Engelke}{Asadi
  et~al.}{2015}]{asadi2015extremes}
Asadi, P., A.~C. Davison, and S.~Engelke (2015).
\newblock Extremes on river networks.
\newblock {\em The Annals of Applied Statistics\/}~{\em 9\/}(4), 2023--2050.

\bibitem[\protect\citeauthoryear{Asadi, Engelke, and Davison}{Asadi
  et~al.}{2018}]{asa2018}
Asadi, P., S.~Engelke, and A.~C. Davison (2018).
\newblock Optimal regionalization of extreme value distributions for flood
  estimation.
\newblock {\em Journal of Hydrology\/}~{\em 556}, 182--193.

\bibitem[\protect\citeauthoryear{Banfi and De~Michele}{Banfi and
  De~Michele}{2022}]{Banfi22}
Banfi, F. and C.~De~Michele (2022).
\newblock Compound ﬂood hazard at lake como, italy, is driven by temporal
  clustering of rainfall events.
\newblock {\em Communications Earth {\&} Environment\/}~{\em 3}.

\bibitem[\protect\citeauthoryear{Barton, Giannakaki, von Waldow, Chevalier,
  Pfahl, and Martius}{Barton et~al.}{2016}]{Barton2016}
Barton, Y., P.~Giannakaki, H.~von Waldow, C.~Chevalier, S.~Pfahl, and
  O.~Martius (2016).
\newblock {Clustering of Regional-Scale Extreme Precipitation Events in
  Southern Switzerland}.
\newblock {\em Monthly Weather Review\/}~{\em 144\/}(1), 347--369.

\bibitem[\protect\citeauthoryear{Bevacqua, De~Michele, Manning, Couasnon,
  Ribeiro, Ramos, Vignotto, Bastos, Blesić, Durante, Hillier, Oliveira, Pinto,
  Ragno, Rivoire, Saunders, van~der Wiel, Wu, Zhang, and Zscheischler}{Bevacqua
  et~al.}{2021}]{Bevacquaetal2021}
Bevacqua, E., C.~De~Michele, C.~Manning, A.~Couasnon, A.~F.~S. Ribeiro, A.~M.
  Ramos, E.~Vignotto, A.~Bastos, S.~Blesić, F.~Durante, J.~Hillier, S.~C.
  Oliveira, J.~G. Pinto, E.~Ragno, P.~Rivoire, K.~Saunders, K.~van~der Wiel,
  W.~Wu, T.~Zhang, and J.~Zscheischler (2021).
\newblock Guidelines for studying diverse types of compound weather and climate
  events.
\newblock {\em Earth's Future\/}~{\em 9\/}(11), e2021EF002340.
\newblock e2021EF002340 2021EF002340.

\bibitem[\protect\citeauthoryear{Bevacqua, Maraun, {Hob{\ae}k Haff}, Widmann,
  and Vrac}{Bevacqua et~al.}{2017}]{Bevacqua2017}
Bevacqua, E., D.~Maraun, I.~{Hob{\ae}k Haff}, M.~Widmann, and M.~Vrac (2017).
\newblock {Multivariate statistical modelling of compound events via
  pair-copula constructions: Analysis of floods in Ravenna (Italy)}.
\newblock {\em Hydrology and Earth System Sciences\/}.

\bibitem[\protect\citeauthoryear{Boers, Bookhagen, Barbosa, Kurths, and
  Marengo}{Boers et~al.}{2014}]{Boers2014natcom}
Boers, N., B.~Bookhagen, H.~Barbosa, J.~Kurths, and J.~Marengo (2014, 10).
\newblock Prediction of extreme floods in the eastern central andes based on a
  complex networks approach.
\newblock {\em Nature communications\/}~{\em 5}, 5199.

\bibitem[\protect\citeauthoryear{Boulaguiem, Zscheischler, Vignotto, van~der
  Wiel, and Engelke}{Boulaguiem et~al.}{2022}]{bou2022}
Boulaguiem, Y., J.~Zscheischler, E.~Vignotto, K.~van~der Wiel, and S.~Engelke
  (2022).
\newblock Modeling and simulating spatial extremes by combining extreme value
  theory with generative adversarial networks.
\newblock {\em Environmental Data Science\/}~{\em 1}, e5.

\bibitem[\protect\citeauthoryear{Chang and Burn}{Chang and
  Burn}{2003}]{Chang2003a}
Chang, S. and D.~H. Burn (2003).
\newblock {Spatial patterns of homogeneous pooling groups for flood frequency
  analysis}.
\newblock {\em Hydrological Sciences\/}~{\em 48\/}(4), 601--618.

\bibitem[\protect\citeauthoryear{Coles, Heffernan, and Tawn}{Coles
  et~al.}{1999}]{col1999}
Coles, S., J.~Heffernan, and J.~Tawn (1999).
\newblock Dependence measures for extreme value analyses.
\newblock {\em Extremes\/}~{\em 2}, 339--365.

\bibitem[\protect\citeauthoryear{Cooley, Naveau, and Poncet}{Cooley
  et~al.}{2006}]{coo2006}
Cooley, D., P.~Naveau, and P.~Poncet (2006).
\newblock {Variograms for spatial max-stable random fields}.
\newblock In P.~Bertail, P.~Soulier, and P.~Doukhan (Eds.), {\em Dependence in
  Probability and Statistics}, Volume 187 of {\em Lecture Notes in Statistics},
  Chapter~17, pp.\  373--390. New York: Springer.

\bibitem[\protect\citeauthoryear{Davison and Gholamrezaee}{Davison and
  Gholamrezaee}{2012}]{Davison12}
Davison, A.~C. and M.~M. Gholamrezaee (2012).
\newblock Geostatistics of extremes.
\newblock {\em Proceedings of the Royal Society A: Mathematical, Physical and
  Engineering Sciences\/}~{\em 468\/}(2138), 581--608.

\bibitem[\protect\citeauthoryear{{De Luca}, Hillier, Wilby, Quinn, and
  Harrigan}{{De Luca} et~al.}{2017}]{DeLuca2017}
{De Luca}, P., J.~K. Hillier, R.~L. Wilby, N.~W. Quinn, and S.~Harrigan (2017).
\newblock Extreme multi-basin flooding linked with extra-tropical cyclones.
\newblock {\em Environmental Research Letters\/}~{\em 12\/}(11), 114009.

\bibitem[\protect\citeauthoryear{Deidda, Rahimi, and De~Michele}{Deidda
  et~al.}{2021}]{deidda2021causes}
Deidda, C., L.~Rahimi, and C.~De~Michele (2021).
\newblock Causes of dependence between extreme floods.
\newblock {\em Environmental Research Letters\/}~{\em 16\/}(8), 084002.

\bibitem[\protect\citeauthoryear{Deuber, Li, Engelke, and Maathuis}{Deuber
  et~al.}{2021}]{Deuber}
Deuber, D., J.~Li, S.~Engelke, and M.~H. Maathuis (2021).
\newblock Estimation and inference of extremal quantile treatment effects for
  heavy-tailed distributions.

\bibitem[\protect\citeauthoryear{Di~Capua, Coumou, Hurk, Donner, Raghavan,
  Vellore, Runge, and Turner}{Di~Capua et~al.}{2020}]{DiCapua20}
Di~Capua, G., D.~Coumou, B.~Hurk, R.~Donner, K.~Raghavan, R.~Vellore, J.~Runge,
  and A.~Turner (2020, 10).
\newblock Dominant patterns of interaction between the tropics and
  mid-latitudes in boreal summer: causal relationships and the role of
  timescales.
\newblock {\em Weather and Climate Dynamics\/}~{\em 1}.

\bibitem[\protect\citeauthoryear{Dixon, Hannaford, and Fry}{Dixon
  et~al.}{2013}]{Dixon2013}
Dixon, H., J.~Hannaford, and M.~J. Fry (2013).
\newblock {The effective management of national hydrometric data: experiences
  from the United Kingdom}.
\newblock {\em Hydrological Sciences Journal\/}~{\em 58\/}(7), 1383--1399.

\bibitem[\protect\citeauthoryear{Durante and Sempi}{Durante and
  Sempi}{2010}]{Durante10}
Durante, F. and C.~Sempi (2010).
\newblock Copula theory: An introduction.
\newblock In P.~Jaworski, F.~Durante, W.~K. H{\"a}rdle, and T.~Rychlik (Eds.),
  {\em Copula Theory and Its Applications}, Berlin, Heidelberg, pp.\  3--31.
  Springer Berlin Heidelberg.

\bibitem[\protect\citeauthoryear{Ebert-Uphoff and Deng}{Ebert-Uphoff and
  Deng}{2012}]{Immecausal}
Ebert-Uphoff, I. and Y.~Deng (2012).
\newblock Causal discovery for climate research using graphical models.
\newblock {\em Journal of Climate\/}~{\em 25\/}(17), 5648 -- 5665.

\bibitem[\protect\citeauthoryear{Engelke and Hitz}{Engelke and
  Hitz}{2020}]{eng2018}
Engelke, S. and A.~S. Hitz (2020).
\newblock Graphical models for extremes (with discussion).
\newblock {\em J. R. Stat. Soc. Ser. B. Stat. Methodol.\/}~{\em 82\/}(4),
  871--932.

\bibitem[\protect\citeauthoryear{Engelke and Ivanovs}{Engelke and
  Ivanovs}{2021}]{eng2021}
Engelke, S. and J.~Ivanovs (2021).
\newblock Sparse structures for multivariate extremes.
\newblock {\em Annual Review of Statistics and Its Application\/}~{\em 8},
  241--270.

\bibitem[\protect\citeauthoryear{Engelke, Malinowski, Kabluchko, and
  Schlather}{Engelke et~al.}{2015}]{Engelke2015}
Engelke, S., A.~Malinowski, Z.~Kabluchko, and M.~Schlather (2015).
\newblock Estimation of {H}\"{u}sler--{R}eiss distributions and
  {B}rown--{R}esnick processes.
\newblock {\em J. R. Stat. Soc. Ser. B. Stat. Methodol.\/}~{\em 77\/}(1),
  239--265.

\bibitem[\protect\citeauthoryear{Engelke and Volgushev}{Engelke and
  Volgushev}{2022}]{eng2020}
Engelke, S. and S.~Volgushev (2022).
\newblock Structure learning for extremal tree models.
\newblock {\em J. R. Stat. Soc. Ser. B Stat. Methodol.\/}.
\newblock Forthcoming.

\bibitem[\protect\citeauthoryear{Fang and Pomeroy}{Fang and
  Pomeroy}{2016}]{Fang2016}
Fang, X. and J.~W. Pomeroy (2016).
\newblock {Impact of antecedent conditions on simulations of a flood in a
  mountain headwater basin}.

\bibitem[\protect\citeauthoryear{Fry}{Fry}{2013}]{Fry2013}
Fry, M. (2013).
\newblock {Hydrological data management systems within a National River Flow
  Archive}.
\newblock pp.\  1--8.

\bibitem[\protect\citeauthoryear{Gagliardini and Gourieroux}{Gagliardini and
  Gourieroux}{2022}]{Gagliardini22}
Gagliardini, P. and C.~Gourieroux (2022).
\newblock Constrained nonparametric dependence with application in finance and
  insurance.

\bibitem[\protect\citeauthoryear{Genest, Kojadinovic, Ne{\v s}lehov{\'a}, and
  Yan}{Genest et~al.}{2011}]{Genest11}
Genest, C., I.~Kojadinovic, J.~Ne{\v s}lehov{\'a}, and J.~Yan (2011).
\newblock {A goodness-of-fit test for bivariate extreme-value copulas}.
\newblock {\em {Bernoulli}\/}~{\em 17\/}(1), 253--275.
\newblock cited By (since 1996)12.

\bibitem[\protect\citeauthoryear{Gnecco, Meinshausen, Peters, and
  Engelke}{Gnecco et~al.}{2021}]{gnecco2021causal}
Gnecco, N., N.~Meinshausen, J.~Peters, and S.~Engelke (2021).
\newblock {Causal discovery in heavy-tailed models}.
\newblock {\em Ann. Statist.\/}~{\em 49\/}(3), 1755 -- 1778.

\bibitem[\protect\citeauthoryear{Gudendorf and Segers}{Gudendorf and
  Segers}{2010}]{seg2010}
Gudendorf, G. and J.~Segers (2010).
\newblock Extreme-value copulas.
\newblock In {\em Copula Theory and Its Applications}, pp.\  127--145.
  Springer.

\bibitem[\protect\citeauthoryear{Hatemi-J}{Hatemi-J}{2012}]{Hatemi12}
Hatemi-J, A. (2012).
\newblock Asymmetric causality tests with an application.
\newblock {\em Empirical Economics - EMPIR ECON\/}~{\em 43}, 1--10.

\bibitem[\protect\citeauthoryear{Hatherley and Alcock}{Hatherley and
  Alcock}{2007}]{Hatherley07}
Hatherley, A. and J.~Alcock (2007).
\newblock Portfolio construction incorporating asymmetric dependence
  structures: a user's guide.
\newblock {\em Accounting \& Finance\/}~{\em 47\/}(3), 447--472.

\bibitem[\protect\citeauthoryear{H{\"u}sler and Reiss}{H{\"u}sler and
  Reiss}{1989}]{HR1989}
H{\"u}sler, J. and R.-D. Reiss (1989).
\newblock Maxima of normal random vectors: between independence and complete
  dependence.
\newblock {\em Stat. Probab. Lett.\/}~{\em 7\/}(4), 283--286.

\bibitem[\protect\citeauthoryear{Kemter, Merz, Marwan, Vorogushyn, and
  Blöschl}{Kemter et~al.}{2020}]{Kemter20}
Kemter, M., B.~Merz, N.~Marwan, S.~Vorogushyn, and G.~Blöschl (2020, 04).
\newblock Joint trends in flood magnitudes and spatial extents across europe.
\newblock {\em Geophysical Research Letters\/}~{\em 47}.

\bibitem[\protect\citeauthoryear{Kim, Johnson, Cifelli, Thorstensen, and
  Chandrasekar}{Kim et~al.}{2019}]{KIM2019100629}
Kim, L.~Johnson, R.~Cifelli, A.~Thorstensen, and V.~Chandrasekar (2019, 12).
\newblock Assessment of antecedent moisture condition on flood frequency: An
  experimental study in napa river basin, ca.
\newblock {\em Journal of Hydrology: Regional Studies\/}~{\em 26}.

\bibitem[\protect\citeauthoryear{Kim and Kim}{Kim and Kim}{2014}]{Kim14}
Kim, D. and J.~Kim (2014).
\newblock Analysis of directional dependence using asymmetric copula-based
  regression models.
\newblock {\em Journal of Statistical Computation and Simulation\/}~{\em
  84\/}(9), 1990--2010.

\bibitem[\protect\citeauthoryear{Larsson and Resnick}{Larsson and
  Resnick}{2012}]{lar2012}
Larsson, M. and S.~I. Resnick (2012).
\newblock Extremal dependence measure and extremogram: the regularly varying
  case.
\newblock {\em Extremes\/}~{\em 15}, 231--256.

\bibitem[\protect\citeauthoryear{Li, Wang, and Duan}{Li et~al.}{2020}]{Li2020}
Li, W., Q.~J. Wang, and Q.~Duan (2020).
\newblock {A variable-correlation model to characterize asymmetric dependence
  for postprocessing short-term precipitation forecasts}.
\newblock {\em Monthly Weather Review\/}~{\em 148\/}(1), 241--257.

\bibitem[\protect\citeauthoryear{Liebscher}{Liebscher}{2008}]{Liebscher08}
Liebscher, E. (2008).
\newblock Construction of asymmetric multivariate copulas.
\newblock {\em Journal of Multivariate Analysis\/}~{\em 99\/}(10), 2234--2250.

\bibitem[\protect\citeauthoryear{Mhalla, Chavez-Demoulin, and Dupuis}{Mhalla
  et~al.}{2019}]{Mhalla2019}
Mhalla, L., V.~Chavez-Demoulin, and D.~J. Dupuis (2019).
\newblock Causal mechanism of extreme river discharges in the upper danube
  basin network.
\newblock {\em arXiv: Applications\/}.

\bibitem[\protect\citeauthoryear{Okimoto}{Okimoto}{2008}]{okimoto2008}
Okimoto, T. (2008).
\newblock New evidence of asymmetric dependence structures in international
  equity markets.
\newblock {\em Journal of Financial and Quantitative Analysis\/}~{\em 43\/}(3),
  787–815.

\bibitem[\protect\citeauthoryear{Ombadi, Nguyen, Sorooshian, and Hsu}{Ombadi
  et~al.}{2020}]{Ombadi2020}
Ombadi, M., P.~Nguyen, S.~Sorooshian, and K.-l. Hsu (2020).
\newblock Evaluation of methods for causal discovery in hydrometeorological
  systems.
\newblock {\em Water Resources Research\/}~{\em 56\/}(7), e2020WR027251.
\newblock e2020WR027251 2020WR027251.

\bibitem[\protect\citeauthoryear{Padoan}{Padoan}{2011}]{PADOAN2011977}
Padoan, S.~A. (2011).
\newblock Multivariate extreme models based on underlying skew-t and
  skew-normal distributions.
\newblock {\em Journal of Multivariate Analysis\/}~{\em 102\/}(5), 977--991.

\bibitem[\protect\citeauthoryear{Padoan}{Padoan}{2013}]{PADOAN20131}
Padoan, S.~A. (2013).
\newblock Extreme dependence models based on event magnitude.
\newblock {\em Journal of Multivariate Analysis\/}~{\em 122}, 1--19.

\bibitem[\protect\citeauthoryear{Pasche and Engelke}{Pasche and
  Engelke}{2022}]{Olivier2022}
Pasche, O.~C. and S.~Engelke (2022).
\newblock Neural networks for extreme quantile regression with an application
  to forecasting of flood risk.

\bibitem[\protect\citeauthoryear{Patton}{Patton}{2004}]{Patton04}
Patton, A.~J. (2004, 01).
\newblock {On the Out-of-Sample Importance of Skewness and Asymmetric
  Dependence for Asset Allocation}.
\newblock {\em Journal of Financial Econometrics\/}~{\em 2\/}(1), 130--168.

\bibitem[\protect\citeauthoryear{Rao}{Rao}{1962}]{rao1962}
Rao, R.~R. (1962).
\newblock Relations between weak and uniform convergence of measures with
  applications.
\newblock {\em The Annals of Mathematical Statistics\/}~{\em 33\/}(2),
  659--680.

\bibitem[\protect\citeauthoryear{Ridder, Pitman, Westra, Do, Bador, Hirsch,
  Evans, Di~Luca, and Zscheischler}{Ridder et~al.}{2020}]{Ridder20}
Ridder, N., A.~Pitman, S.~Westra, H.~Do, M.~Bador, A.~Hirsch, J.~Evans,
  A.~Di~Luca, and J.~Zscheischler (2020, 11).
\newblock Global hotspots for the occurrence of compound events.
\newblock {\em Nature Communications\/}~{\em 11}.

\bibitem[\protect\citeauthoryear{Robson and Reed}{Robson and
  Reed}{1999}]{Robson1999}
Robson, A. and D.~Reed (1999).
\newblock {\em Flood estimation handbook. Vol. 3, Statistical procedures for
  flood frequency estimation}.

\bibitem[\protect\citeauthoryear{Runge, Bathiany, Bollt, Camps-Valls, Coumou,
  Deyle, Glymour, Kretschmer, Mahecha, Mu{\~n}oz-Mar{\'i}, {van Nes}, Peters,
  Quax, Reichstein, Scheffer, Sch{\"o}lkopf, Spirtes, Sugihara, Sun, Zhang, and
  Zscheischler}{Runge et~al.}{2019}]{Runge2019}
Runge, J., S.~Bathiany, E.~Bollt, G.~Camps-Valls, D.~Coumou, E.~Deyle,
  C.~Glymour, M.~Kretschmer, M.~Mahecha, J.~Mu{\~n}oz-Mar{\'i}, E.~{van Nes},
  J.~Peters, R.~Quax, M.~Reichstein, M.~Scheffer, B.~Sch{\"o}lkopf, P.~Spirtes,
  G.~Sugihara, J.~Sun, K.~Zhang, and J.~Zscheischler (2019).
\newblock Inferring causation from time series in earth system sciences.
\newblock {\em Nature Communications\/}~{\em 10\/}(1).

\bibitem[\protect\citeauthoryear{Runge, Nowack, Kretschmer, Flaxman, and
  Sejdinovic}{Runge et~al.}{2019}]{Runge2019b}
Runge, J., P.~Nowack, M.~Kretschmer, S.~Flaxman, and D.~Sejdinovic (2019).
\newblock Detecting and quantifying causal associations in large nonlinear time
  series datasets.
\newblock {\em Science Advances\/}~{\em 5\/}(11), eaau4996.

\bibitem[\protect\citeauthoryear{Salvadori, De~Michele, Kottegoda, and
  Rosso}{Salvadori et~al.}{2007}]{salvadori2007}
Salvadori, G., C.~De~Michele, N.~T. Kottegoda, and R.~Rosso (2007).
\newblock {\em Extremes in nature: an approach using copulas}, Volume~56.
\newblock Springer Science \& Business Media.

\bibitem[\protect\citeauthoryear{Schneeberger, R{\"{o}}ssler, and
  Weingartner}{Schneeberger et~al.}{2018}]{Schneeberger2018}
Schneeberger, K., O.~R{\"{o}}ssler, and R.~Weingartner (2018, 03).
\newblock Spatial patterns of frequent floods in switzerland.
\newblock {\em Hydrological Sciences Journal/Journal des Sciences
  Hydrologiques\/}.

\bibitem[\protect\citeauthoryear{Shimizu, Hoyer, Hyv{\"a}rinen, and
  Kerminen}{Shimizu et~al.}{2006}]{shimizu2006linear}
Shimizu, S., P.~O. Hoyer, A.~Hyv{\"a}rinen, and A.~Kerminen (2006).
\newblock A linear non-gaussian acyclic model for causal discovery.
\newblock {\em Journal of Machine Learning Research\/}~{\em 7\/}(Oct),
  2003--2030.

\bibitem[\protect\citeauthoryear{Stephenson}{Stephenson}{2002}]{Evdpackage}
Stephenson, A.~G. (2002).
\newblock evd: Extreme value distributions.
\newblock {\em R News\/}~{\em 2\/}(2), 0.

\bibitem[\protect\citeauthoryear{Tran, Buck, and Klüppelberg}{Tran
  et~al.}{2021}]{tra2021}
Tran, N.~M., J.~Buck, and C.~Klüppelberg (2021).
\newblock Estimating a directed tree for extremes.

\bibitem[\protect\citeauthoryear{Tsafack}{Tsafack}{2009}]{Tsafack7}
Tsafack, G. (2009).
\newblock Asymmetric dependence implications for extreme risk management.
\newblock {\em The Journal of Derivatives\/}~{\em 17\/}(1), 7--20.

\bibitem[\protect\citeauthoryear{Vignotto, Engelke, and Zscheischler}{Vignotto
  et~al.}{2021}]{vig2021}
Vignotto, E., S.~Engelke, and J.~Zscheischler (2021).
\newblock Clustering bivariate dependencies of compound precipitation and wind
  extremes over {G}reat {B}ritain and {I}reland.
\newblock {\em Weather Clim. Extrem.\/}~{\em 32}, 100318.

\bibitem[\protect\citeauthoryear{Villarini, Smith, Vitolo, and
  Stephenson}{Villarini et~al.}{2013}]{Villarini13}
Villarini, G., J.~A. Smith, R.~Vitolo, and D.~B. Stephenson (2013).
\newblock On the temporal clustering of us floods and its relationship to
  climate teleconnection patterns.
\newblock {\em International Journal of Climatology\/}~{\em 33\/}(3), 629--640.

\bibitem[\protect\citeauthoryear{Ward, Couasnon, Eilander, Haigh, Hendry, Muis,
  Veldkamp, Winsemius, and Wahl}{Ward et~al.}{2018}]{Ward2018}
Ward, P., A.~Couasnon, D.~Eilander, I.~Haigh, A.~Hendry, S.~Muis, T.~I.
  Veldkamp, H.~Winsemius, and T.~Wahl (2018, 07).
\newblock Dependence between high sea-level and high river discharge increases
  flood hazard in global deltas and estuaries.
\newblock {\em Environmental Research Letters\/}~{\em 13}.

\bibitem[\protect\citeauthoryear{Ward, Jongman, Kummu, Dettinger, Weiland, and
  Winsemius}{Ward et~al.}{2014}]{Ward2014}
Ward, P., B.~Jongman, M.~Kummu, M.~D. Dettinger, F.~C. Weiland, and H.~C.
  Winsemius (2014, nov).
\newblock {Strong influence of El Ni{\~{n}}o Southern Oscillation on flood risk
  around the world}.
\newblock {\em Proceedings of the National Academy of Sciences of the United
  States of America\/}~{\em 111\/}(44), 15659--15664.

\bibitem[\protect\citeauthoryear{Wiedermann, Kim, Sungur, and von
  Eye}{Wiedermann et~al.}{2021}]{Wolfgang21}
Wiedermann, W., D.~Kim, E.~A. Sungur, and A.~von Eye (2021).
\newblock Direction dependence in statistical modeling: Methods of analysis.
\newblock {\em International Statistical Review\/}~{\em 89\/}(3), 659--660.

\bibitem[\protect\citeauthoryear{Zhang, Cheng, Jin, Yang, Flerchinger, Chang,
  Bense, Han, and Liang}{Zhang et~al.}{2017}]{Zhang17}
Zhang, Y., G.~Cheng, H.~Jin, D.~Yang, G.~Flerchinger, X.~Chang, V.~Bense,
  X.~Han, and J.~Liang (2017, 04).
\newblock Influences of frozen ground and climate change on the hydrological
  processes in an alpine watershed: A case study in the upstream area of the
  hei'he river, northwest china.
\newblock {\em Permafrost and Periglacial Processes\/}~{\em 28}, 420--432.

\bibitem[\protect\citeauthoryear{Zhang, Gomes, Beer, Neumann, Nackenhorst, and
  Kim}{Zhang et~al.}{2019}]{ZHANG20191960}
Zhang, Y., A.~T. Gomes, M.~Beer, I.~Neumann, U.~Nackenhorst, and C.-W. Kim
  (2019).
\newblock Modeling asymmetric dependences among multivariate soil data for the
  geotechnical analysis – the asymmetric copula approach.
\newblock {\em Soils and Foundations\/}~{\em 59\/}(6), 1960--1979.

\bibitem[\protect\citeauthoryear{Zscheischler, Martius, Westra, Bevacqua,
  Raymond, Horton, Hurk, AghaKouchak, Jézéquel, Mahecha, Maraun, Ramos,
  Ridder, Thiery, and Vignotto}{Zscheischler et~al.}{2020}]{Zscheischler20}
Zscheischler, J., O.~Martius, S.~Westra, E.~Bevacqua, C.~Raymond, R.~Horton,
  B.~Hurk, A.~AghaKouchak, A.~Jézéquel, M.~Mahecha, D.~Maraun, A.~Ramos,
  N.~Ridder, W.~Thiery, and E.~Vignotto (2020).
\newblock A typology of compound weather and climate events.
\newblock {\em Nature Reviews Earth {\&} Environment\/}~{\em 1}, 1--15.

\bibitem[\protect\citeauthoryear{Zscheischler, Naveau, Martius, Engelke, and
  Raible}{Zscheischler et~al.}{2021}]{zsc2021}
Zscheischler, J., P.~Naveau, O.~Martius, S.~Engelke, and C.~Raible (2021).
\newblock Evaluating the dependence structure of compound precipitation and
  wind speed extremes.
\newblock {\em Earth Syst. Dyn.\/}~{\em 12}, 1--16.

\bibitem[\protect\citeauthoryear{Zscheischler and Seneviratne}{Zscheischler and
  Seneviratne}{2017}]{DepZsch17cop}
Zscheischler, J. and S.~Seneviratne (2017, 06).
\newblock Dependence of drivers affects risks associated with compound events.
\newblock {\em Science Advances\/}~{\em 3}, e1700263.

\bibitem[\protect\citeauthoryear{Zscheischler, Westra, Hurk, Seneviratne, Ward,
  Pitman, AghaKouchak, Bresch, Leonard, Wahl, and Zhang}{Zscheischler
  et~al.}{2018}]{Zscheischler18}
Zscheischler, J., S.~Westra, B.~Hurk, S.~Seneviratne, P.~Ward, A.~Pitman,
  A.~AghaKouchak, D.~Bresch, M.~Leonard, T.~Wahl, and X.~Zhang (2018).
\newblock Future climate risk from compound events.
\newblock {\em Nature Climate Change\/}~{\em 8}.

\end{thebibliography}

\appendix

\section{Theoretical results}\label{theory}

We recall that we assume that $X$ and $Y$ have continuous distribution functions $F_X$ and $F_Y$, respectively.
The original Kendall's $\tau$ is defined in~\eqref{OriginalKT} and can be represented by either the joint distribution function $F$ or the copula $C$ of the bivariate random vector $(X,Y)$. Indeed, we have that
\begin{align}\label{alt_rep}
\tau(X,Y) = 4 \mathbb E F(X,Y) - 1 = 4 \mathbb E C(U,V) - 1,
\end{align}
where $(U,V) \sim C$ is the random vector with distribution function $C$ \citep{salvadori2007}.

Since the pre-limit asymmetric tail Kendall's $\tau_{XY}(q)$ is invariant under monotone marginal transformations for any $q \in [0,1]$, this also holds for the limiting version $\tau_{XY}$. We can therefore define normalized variables $\tilde X = 1 / (1- F_X(X))$ and $\tilde Y = 1 / (1- F_Y(Y))$ that have standard Pareto distributions, and note that $\tau_{XY}(q) = \tau_{\tilde X \tilde Y}(q)$ and $\tau_{XY} = \tau_{\tilde X \tilde Y}$.

We concentrate here on models with positive upper tail dependence, also called asymptotical extremal dependence, which is characterized by $\chi_{XY} > 0$.
The assumption of bivariate regular variation is standard in extreme value theory; it is satisfied for instance for all extreme value copulas such as the H\"usler--Reiss distribution and extremal Dirichlet model in Examples~\ref{ex:HR} and~\ref{ex:dirichlet}, respectively.
It implies that the exceedances of $(X,Y)$ over a high threshold converge to a limit distribution. Indeed, for any sequence $q_n \to 1$ as $n\to \infty$, define the threshold $u_n = F_X^{-1}(q_n) $. Then we have the convergence in distribution 
\begin{align}\label{extremal_function}
     (\tilde X / u_n, \tilde Y / u_n) \mid X > u_n  \Rightarrow (P, PW^1_2),
\end{align} 
where $P$ has a standard Pareto distribution with $\mathbb P( P \leq x) = 1 - 1/x$, $x \geq 1$, and $W^1_2\geq 0$ is a non-negative random variable with $\mathbb E W^1_2 = 1$ called the extremal function of variable $Y$ relative to $X$. We assume in the sequel that $W^1_2$ possesses a Lebesgue density on $(0,\infty)$ and therefore is strictly positive. By exchanging the roles of $X$ and $Y$, we can also define the extremal function $W^2_1$, which we also assume to be strictly positive almost surely.

\begin{theorem}\label{thm_tau_extreme}
    Let $(X,Y)$ be multivariate regular varying with extremal functions $W^1_2$ and $W^2_1$. Then the asymmetric tail Kendall's $\tau$ be expressed as
\begin{align}\label{tau_extreme}
\tau_{XY} = 2\mathbb E \min(1,\tilde W^1_2 / W^1_2) - 1, \qquad 
\tau_{YX} = 2\mathbb E \min(1,\tilde W^2_1 / W^2_1) - 1,
\end{align}
$\tilde W^1_2$ and $\tilde W^2_1$ are independent copies of $W^1_2$ and $W^2_1$, respectively.
\end{theorem}

\begin{proof}

Define the random vectors $(W_n, Z_n)$ and $(W,Z)$ as the conditional random vector on the left-hand side and the limiting random vector on the right-hand side of~\eqref{extremal_function}, respectively, and denote by $G_n$ and $G$ their bivariate distribution functions. We recall that since $(W_n, Z_n)$ converges weakly (in distribution) to $(W,Z)$  as $n\to \infty$, and since $G$ is absolutely continuous with respect to Lebesgue measure, we have the uniform convergence \citep[][Theorem 4.2]{rao1962}
\begin{align}\label{unif_conv}
\lim_{n\to \infty} \sup_{x,y \in [0,\infty] \times \mathbb R} | G_n(x,y) - G(x,y) | =0.
\end{align}
To show convergence of asymmetric tail Kendall's $\tau$ we first note that
\[\tau_{XY}(q_n) = \tau_{\tilde X \tilde Y}(q_n) = \tau(W_n Z_n),\]
where the second equations holds since the conditioning $X > F_X^{-1}(q_n)$ is already included in the definition of $(W_n,Z_n)$.
Using the representation~\eqref{alt_rep} of the original Kendall's $\tau$, we can rewrite
\begin{align*}
     \tau_{W_n Z_n} = 4 \mathbb E G_n(W_n,Z_n) - 1. 
\end{align*} 
In order to show that the expectation converges to $\mathbb E G(W,Z)$ as $n\to \infty$, we observe
\begin{align*}
     |\mathbb E G_n(W_n,Z_n) - \mathbb E G(W,Z) | \leq \mathbb E |G_n(W_n,Z_n) - G(W_n,Z_n)|   + |\mathbb E G(W_n,Z_n) - \mathbb E G(W,Z) |.
\end{align*} 
The first term can be bounded by 
\begin{align*}
   \mathbb E |G_n(W_n,Z_n) -  G(W_n,Z_n)|  \leq \sup_{x,y \in [0,\infty] \times \mathbb R} | G_n(x,y) - G(x,y) | 
\end{align*} 
by the uniform convergence~\eqref{unif_conv} it converges to 0. The second term converges also to 0 since $G$ is a continuous, bounded function and by weak convergence $\lim_{n\to \infty} \mathbb E G(W_n,Z_n) = \mathbb E G(W,Z)$.
Consequently,
\[ \tau_{XY} = \lim_{n \to \infty} \tau_{XY}(q_n) = 4 \mathbb E G(W,Z) - 1 = \tau(P, P W^1_2).\]

We will now derive a closed form expression for the asymmetric tail Kendall's $\tau(P, P W^1_2)$. Recall that $(W,Z) = (P, PW^1_2)$. Let $w$ be the density of $W^1_2$ and recall that the density of $P$ is $x^{-2}$ for all $x \geq 1$. We compute (all integrals are over $(0,\infty)$)
\begin{align*}
  \mathbb E G(P,PW^1_2) &= \mathbb E \int \int \einsfun \{ x \geq 1, x \leq P, xy \leq PW^1_2\} x^{-2} w(y) \D x \D y \\
  & = \int \int \int \int \einsfun \{x \geq 1, \tilde x \geq 1, x \leq \tilde x, xy \leq \tilde x \tilde y \} x^{-2} \tilde x^{-2}  w(y) \tilde w(y)\D x \D y \D \tilde x \D \tilde y \\
   & = \int \int \int \int \einsfun \{z \tilde x \geq 1, \tilde x \geq 1, z \leq 1, z  \leq \tilde y / y \} z^{-2} \tilde x^{-3}  w(y) \tilde w(y)\D x \D y \D \tilde x \D \tilde y \\
    & = \frac12 \int \int \int \einsfun \{z \leq \min(1,\tilde y / y )\}  w(y) \tilde w(y)\D x \D y \D \tilde y \\
    & = \frac12  \mathbb E \min(1,\tilde W^1_2 / W^1_2),
\end{align*}
where $\tilde W^1_2$ is an independent copy of $W^1_2$. For the third equality we used integration by substitution with $x = z \tilde x$ and $\D x = \tilde x \D z$, and $\int_{1/z}^\infty \tilde x^{-3} \D \tilde x = z^2/2$. We therefore obtain the final result that the asymmetric tail Kendall's $\tau$ coefficient $\tau_{XY}$ and, by exchanging the roles of $X$ and $Y$, also $\tau_{YX}$, can be expressed in terms of the extremal functions:
\begin{align*}
\tau_{XY} = 2\mathbb E \min(1,\tilde W^1_2 / W^1_2) - 1, \qquad 
\tau_{YX} = 2\mathbb E \min(1,\tilde W^2_1 / W^2_1) - 1,
\end{align*}
where $W^2_1$ is the extremal function of variable $X$ relative to $Y$. Importantly, there is duality between the extremal functions 
\[ \mathbb P( W^2_1 \leq x) = \mathbb E( \einsfun\{1/W^1_2\leq x\} W^1_2), \qquad x \geq 0.\]
\end{proof}

Representation~\eqref{tau_extreme} is extremely useful for several reasons. Firstly, the general formula for Kendall's $\tau$, which involves a bivariate distribution function and an expectation over a bivariate random vector, has been simplified to a simple univariate (the $\min$) function applied to two independent univariate random variables. Secondly, for many existing models the distribution of $W^1_2$ is known explicitely and therefore asymmetric tail Kendall's $\tau$ can be computed easily.

We show this at the example of the popular H\"usler--Reiss distribution \citep{HR1989, Engelke2015}, which can be seen as the Gaussian distributions in the world of extremes. In the bivariate case, this distribution class is characterized by a parameter $\Gamma \in (0,\infty)$, where for $\Gamma \to 0$ the dependence structure approaches complete dependence, and for $\Gamma \to \infty$ independence. For this symmetric distribution, the extremal functions are distributed identically according to log-normal distribution, such that $\log W^1_2, \log W^2_1 \sim \mathcal N(-\Gamma/2, \Gamma)$. We can therefore compute
\begin{align*}
\tau_{XY} = \tau_{YX} = 2\mathbb E \min(1,\exp\{\tilde N - N\}) - 1,
\end{align*}
where $N$ and $\tilde N$ are independent random variables with distribution $\mathcal N(-\Gamma/2, \Gamma)$. Because of independence, $\tilde N - N \sim \mathcal N(0,2\Gamma)$. Let $\phi$ and $\Phi$ denote the density and distribution function of a standard normal distribution, respectively. Then
\begin{align*}
    \mathbb E \min(1,\exp\{\tilde N - N\}) &= \frac12 + \int_{-\infty}^0 \exp\{ \sqrt{2\Gamma} y\} \phi(y) \D y \\
    &= \frac12 +  \frac12 \exp(\Gamma) \Phi(-\sqrt{2\Gamma})/ \Phi(0)\\
    & = \frac12 + \exp(\Gamma) \{1 - \Phi(\sqrt{2\Gamma})\}.
\end{align*}
Therefore we obtain 
\begin{align*}
    \tau_{XY} = \tau_{YX} = 2\exp(\Gamma) \{1 - \Phi(\sqrt{2\Gamma})\}.
\end{align*}
Note that, as expected, $\lim_{\lambda \to 0} \tau_{XY} = 1$ and $\lim_{\lambda \to \infty} \tau_{XY} = 0$, since $2\exp(\Gamma) \{1 - \Phi(\sqrt{2\Gamma})\} \sim O(\Phi(\sqrt{2\Gamma})/ \phi(\sqrt{2\Gamma})) \sim O(1/ \sqrt{2\Gamma})$ as $\Gamma\to \infty$.
We can compare this to the widely used extremal correlation for this distribution, which is given by 
\[ \chi = 2\{1 - \Phi(\sqrt{\Gamma} / 2)\}. \]

\section*{Open Research Section}

The data used for the hydrological application are available at https://nrfa.ceh.ac.uk/.

\section*{Acknowledgments}
This research has been also supported by the "Damocles" COST ACTION CA17109 "Understanding and modeling compound climate and weather events" through a Short Term Scientific Mission at the first author at University of Geneva.

\end{document}